\documentclass{article}

     \PassOptionsToPackage{numbers, compress}{natbib}

     \usepackage[preprint]{neurips_2025}

\usepackage[utf8]{inputenc} 
\usepackage[T1]{fontenc}    
\usepackage{hyperref}       
\usepackage{url}            
\usepackage{booktabs}       
\usepackage{amsfonts}       
\usepackage{nicefrac}       
\usepackage{microtype}      
\usepackage{xcolor}         

\usepackage{amsmath,amsthm}
\usepackage[ruled,vlined]{algorithm2e}
\usepackage{bbm}
\usepackage{graphicx}
\renewcommand{\Vec}[1]{\mbox{\boldmath $#1$}}
\newcommand{\dtv}{d_{\mathrm{TV}}}
\newcommand{\dkl}{D_{\mathrm{KL}}}
\newcommand{\rank}{\mathrm{rank}}
\newcommand{\conv}{\mathrm{conv}}

\newcommand{\Order}{\mathrm{O}}
\newcommand{\mys}{s}

\newtheorem{theorem}{Theorem}
\newtheorem{lemma}[theorem]{Lemma}

\newtheorem{proposition}[theorem]{Proposition}
\newtheorem{condition}{Condition}

\allowdisplaybreaks

\title{
 Sample Complexity of Identifying the Nonredundancy of Nontransitive Games in Dueling Bandits
 }

\author{%
  Shang Lu\\
  Graduate School of ISEE\\
  Kyushu University\\
  Fukuoka, Japan 819-0395 \\
  \texttt{lushang630@gmail.com} \\
   \And
  Shuji Kijima\thanks{This work is partly supported by JSPS KAKENHI Grant Numbers JP23K21645 and 
  JST ERATO Grant Number JPMJER2301, Japan.} \\
  Department of Data Science\\
  Shiga University\\
  Hikone, Japan 522-8522 \\
  \texttt{shuji-kijima@biwako.shiga-u.ac.jp} \\
}

\begin{document}

\maketitle

\begin{abstract}
 Dueling bandit is a variant of the Multi-armed bandit 
    to learn the binary relation by comparisons. 
 Most work on the dueling bandit 
   has targeted transitive relations,  
     that is, totally/partially ordered sets, 
   or assumed at least the existence of a champion 
     such as Condorcet winner and Copeland winner. 
 This work develops an analysis of 
   dueling bandits for \emph{non-transitive} relations.  
 Jan-ken (a.k.a.\ rock-paper-scissors) is a typical example of a non-transitive relation. 
 It is known that  
   a rational player chooses one of three items uniformly at random, 
  which is known to be Nash equilibrium in game theory.  
 Interestingly,  
   any variant of Jan-ken with four items (e.g., rock, paper, scissors, and well) 
   contains at least one useless item, which is never selected by a rational player. 
 This work investigates a dueling bandit problem 
   to identify whether all $n$ items are indispensable in a given win-lose relation. 
 Then, we provide upper and lower bounds of the sample complexity of the identification problem 
  in terms of the determinant of $A$ and a solution of $\Vec{x}^{\top} A = \Vec{0}^{\top}$ 
   where $A$ is an $n \times n$ pay-off matrix that every duel follows. 

\end{abstract}

\section{Introduction}\label{sec:intro}
\paragraph{Dueling bandits}
 The stochastic bandit is an online reinforcement learning model:  
    the learner repeats rounds of choosing one out of $n$-arms and receiving a stochastic reward 
     to maximize the total sum of rewards or to find the best arm. 
  Regret minimization and 
    the sample complexity of the best arm identification are significant topics
     \cite{Robbins52,LR85,Auer02,MT04,EMM06,ABM10,GC11,BC12,Urvoy13,JN14,GK16,KCG16}. 
  It has been extensively investigated 
    in machine learning, optimization, probability theory, etc., and 
    has many applications in the real world, 
    such as recommendation systems.  

 The \emph{Dueling bandit} problem, 
     introduced by Yue and Joachims \cite{YJ09}, 
    is a variant of the stochastic bandit 
   regarding the \emph{binary relation} of the rewards of arms instead of the reward itself. 
 The learner 
    repeats rounds of choosing \emph{a pair} of arms from $n$-arms and 
      receiving a stochastic result, indicating which arm provides more reward 
    to find the best arm, such as the \emph{Condorcet winner} and the \emph{Copeland winner}. 
 It is motivated by the real world, where 
   relative evaluations are more natural or frequent 
   than absolute evaluations employed in classical bandit problems.  
 In this context, 
   most work assumes the win-lose relation 
     to be transitive so that it is total or partial order 
   or to allow a champion, i.e., Condorcet or Copeland winners,  at least \cite{YJ09,YBKJ12,Urvoy13,Zoghi14,Komiyama15,Zoghi15,Komiyama16}. 
 This work focuses on dueling bandits 
    for a \emph{nontransitive} win-lose relation, 
   where the best arm no longer exists.

\begin{figure}[tbp]
  \centering
\begin{tabular}{cccc}
  \includegraphics[height=16mm]{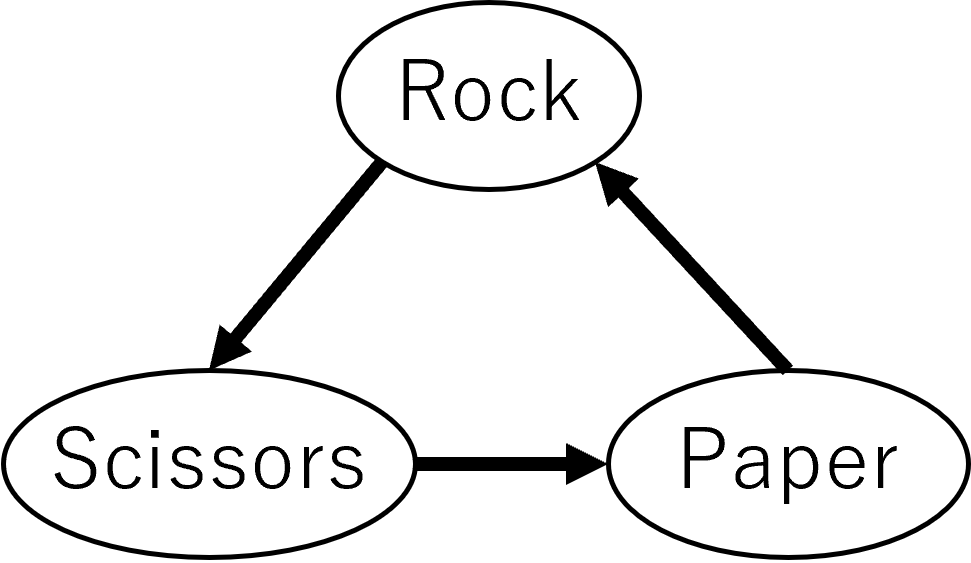} &
  \includegraphics[height=16mm]{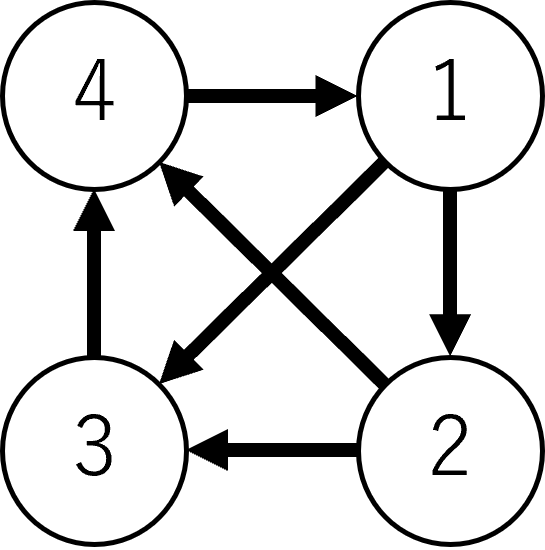} &
  \includegraphics[height=16mm]{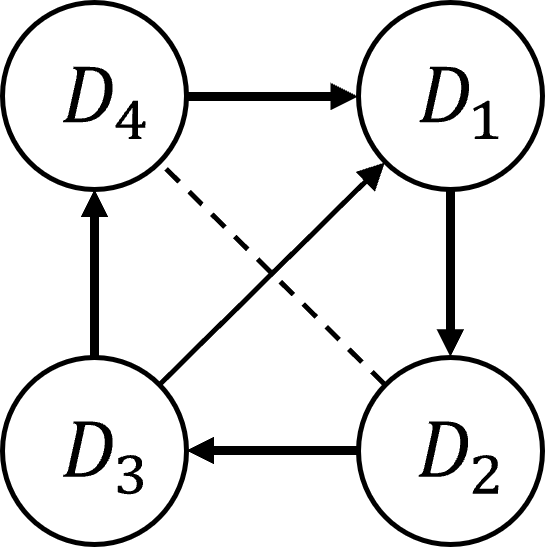} & 
  \includegraphics[height=16mm]{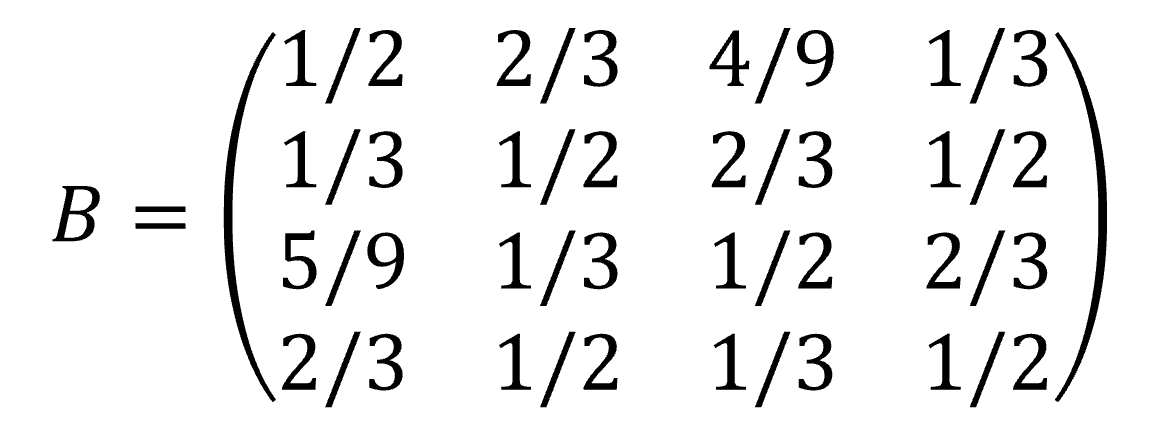}  \\
  (a) Jan-ken & (b) Extended Jan-ken & (c) Effron's dice & (d) Win. prob. matrix
\end{tabular}
  \caption{Examples of nontransitive win-lose relation.}\label{fig:1}
\end{figure}
\paragraph{Nontransitive win-lose relation} 
 Jan-ken, a.k.a. rock-paper-scissors, 
  is a game consisting of three items: rock, paper, and scissors.  
 Rock beats scissors, scissors beat paper, and paper beats rock.
 This is the simplest example of a \emph{nontransitive} relation. 
 Clearly, there is no strongest item (see Figure~\ref{fig:1}~(a)). 

 Similarly, we may consider an extended Jan-ken with \emph{four} items, say\footnote{
  Let $[n]$ denote $\{1,\ldots,n\}$ for a positive integer $n$. 
 } $[4]=\{1,2,3,4\}$. 
 For instance, let $1$ beats $2$ and $3$, $2$ beats $3$ and $4$, $3$ beats $4$, and $4$ beats $1$.  
 This variant also does not have the strongest item, but  
  we can observe that item $3$ is useless 
   because both $2$ and $3$ lose to $1$ and beats $4$, but $2$ beats $3$ (see Figure~\ref{fig:1} (b)). 
 Interestingly, 
  it is known that any win-lose relation on the set $[4]$ contains a useless move 
   unless allowing a tie-break between distinct $i,j$. 

 Another example is nontransitive dice.  
 Efron's dice is a set of four dice $D_1=(0,0,4,4,4,4)$, $D_2=(3,3,3,3,3,3)$, $D_3=(2,2,2,2,6,6)$,  $D_4=(1,1,1,5,5,5)$. 
 The win-lose relation is \emph{stochastic}; 
   roll $D_i$ and $D_j$, and 
   then  $D_i$ beats $D_j$ 
   if the cast of $D_i$ is larger than that of $D_j$.  
 Let $p_{ij}$ denote the probability that the cast of $D_i$ is greater than that of $D_j$. 
 We can observe that $p_{12} = p_{23} = p_{34} = p_{41} = 2/3$, and 
  the stochastic win-lose relation is nontransitive (see Figure~\ref{fig:1}~(c)). 
 Such a dice set is called a nontransitive dice~\cite{Gardner74,Savage94,Schaefer17a,Schaefer17b,LK22,Kim24}. 
 We can find many nontransitive games in the real world, and 
 the nontransitivity makes games nontrivial.

\paragraph{Game theory}
 It is appropriate to follow the terminology of game theory for further discussion. 
 A \emph{two-player zero-sum} game is characterized 
   by a \emph{pay-off matrix} $A = (a_{ij}) \in \mathbb{R}^{m \times n}$
   where $m$ and $n$ respectively represent the numbers of possible moves of row and column players; 
   if the row player selects move $i$ and the column player selects move $j$, then
   the row player gains a profit\footnote{
     Note that $a_{ij}$ can be negative. 
   } of $a_{ij}$ and the column player loses $a_{ij}$, i.e., gains $-a_{ij}$.

 We say $\Vec{x} = (x_1,\ldots,x_n)^{\top} \in \mathbb{R}^n$ is a \emph{mixed strategy} (or simply \emph{strategy}) 
   if $x_i \geq 0$ for $i=1,\ldots,n$ and $\sum_{i=1}^n x_i =1$ hold, 
   where $\Vec{x}$ represents the selection probability of moves $[n]$. 
 If a strategy 
  $\Vec{x}$ satisfies\footnote{
 For a pair of vectors 
   $\Vec{u} = (u_1,\ldots,u_n)^{\top} \in \mathbb{R}^n$ and 
   $\Vec{v} = (v_1,\ldots,v_n)^{\top} \in \mathbb{R}^n$,  
  let $\Vec{u} \geq \Vec{v}$ (resp.\ $\Vec{u} > \Vec{v}$) denote 
    that $u_i \geq v_i$ (resp.\ $u_i > v_i$) holds for any $i = 1,\ldots,n$, 
    $\Vec{u}^{\top} \geq \Vec{v}^{\top}$ and $\Vec{u}^{\top} > \Vec{v}^{\top}$ as well. 
  }  $\Vec{x} > \Vec{0}$, 
  then we say $\Vec{x}$ is \emph{completely} mixed, 
  where $\Vec{0}$ denotes the zero vector\footnote{
   We briefly mention to another special case:  
   a strategy $\Vec{x}$ is \emph{pure} if there exists $i \in \{1,\ldots,n\}$ such that $x_i = 1$. 
  A pure strategy is not the target of this work and we omit the detail.  
}.
 A row strategy $\Vec{x} \in \mathbb{R}^m$ is \emph{$v$-good} for $v \in \mathbb{R}$
   if $\Vec{x}^{\top} A \geq v \Vec{1}^{\top}$ holds
     where $\Vec{1}$ denotes the all one vector. 
 If $\Vec{x}$ is a $v$-good row strategy 
  then the row player's expected gain satisfies $\Vec{x}^{\top} A \Vec{y} \geq v$ 
    for any column strategy $\Vec{y} \in \mathbb{R}^n$, 
  which means that the row player gains at least $v$ in expectation for any column player's strategy.  
 Similarly, 
  a column strategy $\Vec{y}$ is $v'$-good 
   if $A \Vec{y} \leq v' \Vec{1}$ holds.  
  If $\Vec{y}$ is a $v'$-good, 
  then the expected loss of the row player ($=$ expected gain of the column player) 
   satisfies $\Vec{x}^{\top} A \Vec{y} \leq v'$ for any row strategy $\Vec{x}$. 
 It is known for any $A$ that 
   any pair of $v$-good row strategy and $v'$-good column strategy satisfy 
     $v \leq v'$ by the weak duality theorem of linear programming, and 
   $v=v'$ exists by the strong duality (see e.g., \cite{Matousek,von1947theory}). 
 A \emph{Nash equilibrium} 
  is a pair of $v$-good strategies $\Vec{x}$ and $\Vec{y}$,  
  which means that the row player (resp.\ column player) cannot 
   increase (resp.\ decrease) her expected gain (resp.\ loss) from $v$ (resp.\ $v$)
   if the other player is rational. 
 We call $v$ the \emph{game value} of $A$.

 A two-player zero-sum game is \emph{symmetric}  if $A$ is skew-symmetric, i.e., $A^{\top} = -A$ holds, 
   which is the target of this paper. 
 For instance, 
  the pay-off matrix of Jan-ken is given by 
\begin{align*}
  A = \begin{pmatrix}0 & 1 & -1\\-1 & 0 & 1\\1 & -1 &0 \end{pmatrix}
\end{align*}
  which is skew-symmetric. 
 It is not difficult to see that if $\Vec{x}$ is a $v$-good row strategy, it is also a $v$-good column strategy. 
 Thus, it is well-known for two-player zero-sum symmetric games 
   that the pair of $v$-good strategies $\Vec{x}$ and $\Vec{x}$ is a Nash equilibrium, and hence 
  the game value $v$ is zero. 
 We can observe that (the pair of) $(1/3,1/3,1/3)$ is a unique Nash equilibrium of Jan-ken. 

 For another example\footnote{
  We just mention the name of \emph{Colonel Blotto game} for another example (cf., \cite{Matousek}).
 }, 
 the winning probability matrix $B=(b_{ij})$ of Effron's dice is given in Figure~\ref{fig:1}~(d). 
 Let $A = (b_{ij} - \frac{1}{2})$, then $A$ is skew-symmetric, 
  meaning that Effron's dice is essentially regarded as a symmetric game:
 in fact,  its expected gain is 
  $\Vec{x}^{\top}B\Vec{y} 
   = \Vec{x}^{\top}A\Vec{y} + \Vec{x}\frac{1}{2}\mathbbm{1}\Vec{y} 
   = \Vec{x}A\Vec{y}+\frac{1}{2}$, 
   where $\mathbbm{1}$ denotes the $4 \times 4$ all one matrix.  
 Such a game is called a \emph{constant-sum game}. 
 Rump \cite{Rump01} showed that 
  the Nash equilibria of Effron's dice  
   form a line segment between $\Vec{x}= (0,1/2,0,1/2)$ and $\Vec{x}' = (3/7,1/7,3/7,0)$.  
 This means that 
  dice $D_1$ and $D_3$ are no longer used by players in the rational strategy $\Vec{x}$, and 
  neither is $D_4$ in $\Vec{x}'$.

\paragraph{Completely mixed Nash}
 We say a game $A$ is \emph{non-redundant} (or completely mixed) 
   if any Nash equilibrium of $A$ is completely mixed, 
   meaning that all moves are indispensable between rational players. 
 Kaplansky~\cite{Kaplansky45} 
  proved that 
   a game $A\in \mathbb{R}^{m \times n}$ is completely mixed if and only if 
   (1) $A$ is square (i.e., $m=n$) and has rank $n-1$, and 
   (2) all cofactors are different from zero and have the same sign (cf Thm. 5 in \cite{Kaplansky45}). 
 He also pointed out the following facts. 
\begin{theorem}[cf.\ Thm.~5 and Sec.~4 in \cite{Kaplansky45}]\label{thm:kap1}
 Let $A$ be a real $n \times n$ skew-symmetric matrix for $n \geq 2$. 
 Then, $A$ is completely mixed (i.e., non-redundant) only when $n$ is odd. 
 When $n$ is odd, $A$ is non-redundant if and only if 
  $\rank(A) = n-1$ and 
  $\Vec{x}^{\top} A = \Vec{0}^{\top}$ has a solution
   $\Vec{x} > \Vec{0}$. 
\end{theorem}
 The former claim comes from 
   the fact that the determinant of any $k \times k$ skew-symmetric matrix is zero for any even $k$. 
 After 50 years, he in \cite{Kaplansky95} gave a proof of the following theorem 
   which characterizes skew-symmetric $A$ being non-redundant in terms of the principal Pfaffians of $A$, 
    which was already proved for $n=3$ in  \cite{Kaplansky45}. 
\begin{theorem}[Thm.~1 in \cite{Kaplansky95}]\label{thm:kap2}
 Let $A$ be a real $n \times n$ skew-symmetric matrix for an odd $n \geq 3$. 
 Then, $A$ is completely mixed if and only if the principle Pfaffians $p_1,\ldots,p_n$ of $A$ 
 are all nonzero and alternate in sign. In that case the unique good strategy for each player is proportional to 
 $\Vec{p} = (p_1,-p_2,\ldots,(-1)^{n-1}p_n)$.\footnote{
 Let $M_k = (m_{ij})$ be the $(n-1) \times (n-1)$ submatrix of $A$ formed by deleting the $k$-th row and column,  
 then the $k$-th principal Pfaffian is given by 
 $p_k = \frac{1}{2^n n!}\sum_{\sigma \in S_{2n}}\mathrm{sgn}(\sigma) \prod_{i=1}^n m_{\sigma(2i-1),\sigma(2i)}$
 where $S_{2n}$ denotes the symmetric group of degree $2n$. 
 Note that $(-1)^{i+j}p_i p_j$ is equal to the $(i,j)$-cofactor of $A$ (Lem.~1 in \cite{Kaplansky95}), 
  and hence $\Vec{p}^{\top}A = \Vec{0}^{\top}$ holds since $\det(A)=0$
    by a standard argument of linear algebra (see also \cite{Kaplansky45}). 
} 
\end{theorem}

\paragraph{Problem and contribution}
 While most work on dueling bandits is concerned with the ``strongest'' item, 
   this paper focuses on dueling bandits for \emph{nontransitive relations}. 
 We are concerned with the sample complexity of the dueling bandit 
   to identify whether a given matrix $A \in \mathbb{R}^{n \times n}$ is completely mixed. 
 Since the answer is always no for any even $n$ due to Kaplansky \cite{Kaplansky45,Kaplansky95}, 
   this paper is concerned with only odd $n$.  
  In fact, the case of $n=3$ is easy, and we are mainly involved in the case of odd $n \geq 5$. 

 Our dueling bandit setting essentially follows the work of Maiti et al.~\cite{maiti2023instance}, described as follows: 
 Given an unknown $n \times n$ skew-symmetric matrix $A = (a_{ij})$ for an odd $n \geq 3$. 
 We assume that every $a_{ij}$ is finite, namely $a_{ij} \in [-1,1]$, for simplicity of descriptions. 
 A learner lets all pairs of $\{i,j\} \in \binom{[n]}{2}$ duel in a round 
  where $\binom{[n]}{2}$ denotes the set of all pairs of elements in $[n]$, and 
  receives results $X_{ij}$ where each result independently follows 
    a \emph{sub-Gaussian} distribution with mean $a_{ij}$ and variance at most 1. 
 Repeating rounds, the learner decides whether the given matrix $A$ is non-redundant. 
 We refer to the number of rounds required for the decision as sample complexity. 

 We give an upper bound of the sample complexity  
   $\Order\!\left(\frac{\varphi(A)^2}{\max\{\alpha^2,\pi_{\min}^2\}} \log \frac{n}{\delta}\right)$ 
    of an $(\alpha,\delta)$-PAC algorithm for the problem 
  where $\varphi(A)$ is a parameter intuitively related to $1/\det(A)$, 
    $\pi_{\min} = \min_{i \in [n]} \pi_i$ for the unique Nash equilibrium $\Vec{\pi} =(\pi_1,\ldots,\pi_n)$ of non-redundant $A$, 
    $\alpha$ is a prescribed margin parameter, and 
    $\delta$ is the confidence level. 
 We also give lower bounds 
  $\Omega\!\left(\frac{1}{\alpha^2}\log \frac{1}{\delta}\right)$ for any $n$ and 
  $\Omega\!\left(\varphi(A)^2 \log \frac{1}{\delta}\right)$ for each $n = 5, 7, \ldots,19$, 
  which provide  
  $\Omega\!\left(\max\{\frac{1}{\alpha^2}, \varphi(A)^2 \}\log \frac{1}{\delta}\right)$ for each $n = 5, 7, \ldots,19$. 
 Though there is some gap between the upper and lower bounds, 
 our result suggests 
  the possible involvement of $\varphi(A)$ to 
    the sample complexity in the identification of non-redundancy in the game. 
 As far as we know, 
  this is the first result of the problem. 

The work of Maiti et al.~\cite{maiti2023instance} is closely related. 
They investigated 
  the sample complexity of finding an $\epsilon$-Nash equilibrium (see Section \ref{sec:term} for definition) 
   of $2 \times 2$ zero-sum games.  
 They derived an instance-dependent lower bound on the sample complexity.
 This paper focuses on the non-redundancy of games for three or more moves, and 
  our result is incomparable with \cite{maiti2023instance}.

\paragraph{Other related work}
 The sample complexity is a central issue in stochastic multi-armed bandits. 
 \citet{MT04} gave the lower bound of best arm identification 
   $\Omega((n/\epsilon^2) \log (1/\delta))$. 
 \citet{EMM06} 
   proved that the upper bound of the sample complexity matches the lower bound by \cite{MT04} 
  by giving the successive elimination algorithm.

 \citet{YJ09} introduced the dueling bandit framework 
  featuring pairwise comparisons as actions. 
 \citet{YBKJ12} gave a regret lower bound of $\Omega(n\log{T})$ 
   for the $n$-armed dueling bandit problem of $T$ rounds 
   assuming strongly stochastic transitivity. 
 \citet{Urvoy13} proposed the SAVAGE algorithm 
  and gave an instance-dependent upper bound of the sample complexity 
  $\sum_{i=1}^n \Order\!\left(\frac{1}{\Delta_i^2}\log \frac{n}{\delta \Delta_i^2} \right)$ 
  where $\Delta_i$ is the local independence parameter. 
 \citet{Zoghi14} gave an upper bound of regret bound, 
  which matches the lower bound by \cite{YBKJ12} 
  assuming the Condorcet winner. 
 \citet{Komiyama15} further analyzed this lower bound and 
   determined the optimal constant factor for models adhering to the Condorcet assumption and 
   assuming the Condorcet winner arm. 
 \citet{Zoghi15} investigated regret minimization
   concerning the Copeland winner. 
 \citet{Komiyama16} gave an asymptotic regret lower bound 
   based on the minimum amount of exploration for identifying a Copeland winner.

 There are several works on dueling bandits from the viewpoint of game theory. 
 \citet{Ailon14} gave some reduction algorithms from dueling bandit to multi-armed bandit. 
 \citet{Zhou17} initiated the study of identifying the pure strategy Nash equilibrium (PSNE) 
   of a two-player zero-sum matrix game with stochastic results 
  and gave a lower bound of the sample complexity $\Omega(H_1\log(1/\delta))$ where 
   $H_1 = \sum_{i \neq i_*} \frac{1}{(A_{i_*,j_*}-A_{i,j_*})^2} + \sum_{j \neq j_*} \frac{1}{(A_{i_*,j_*}-A_{i_*,j})^2}$ 
   for the PSNE $(i_*,j_*)$ of $A$. 
 \citet{maiti2023instance} investigated the sample complexity of 
   identifying an $\epsilon$-Nash Equilibrium in a two-player zero-sum $2 \times n$ game 
   and provided near-optimal instance-dependent bounds, including 
   the gaps between the entries of the matrix, and the difference between the value of the game and reward received from playing a sub-optimal row. 
 \citet{maiti2023log} extended the techniques of \cite{maiti2023instance} 
   to identify the support of the Nash equilibrium in $m \times n$ games, 
   but the bounds are sub-optimal.
 \citet{maiti24a} investigated the sample complexity of 
   identifying the PSNE in $m \times n$ games and  
   designed a near-optimal algorithm whose sample complexity matches the lower bound by \cite{Zhou17}, up to log factors. 
 \citet{ILTW25} studied a more general class of two-player zero-sum games and 
   derived a regret upper bound $\Order(\sqrt{T} + \frac{m+n}{c} \log T)$
   where $c$ is a game-specific constant dependent on the pay-off structure, 
  and gave a regret lower bound of $\Omega(\log T)$ in cases where the game admits a PSNE.
  While most work focuses on PSNE, 
 \citet{Dudik15} discussed the mixed Nash strategy, as the name of von Neumann winner, 
  of a two-player zero-sum game and gave three algorithms for regret minimization. 

\section{Preliminary}\label{sec:prelim}
\subsection{Terminology}\label{sec:term}
 This paper is concerned with a \emph{two-player zero-sum symmetric game}, 
  which is characterized by an $n \times n$ skew-symmetric matrix $A \in [-1,1]^{n \times n}$. 
 For convenience, we define 
\begin{align*}
 \mathcal{S}_n &= \left\{ \Vec{x} \in \mathbb{R}^n \mid \textstyle \sum_{i=1}^n x_i = 1 \right\}, \\
 \mathcal{S}_n^+ &= \left\{ \Vec{x} \in \mathcal{S}_n \mid  \Vec{x} \geq \Vec{0} \right\}, \quad \mbox{and} \\
 \mathcal{S}_n^{++} &= \left\{ \Vec{x} \in \mathcal{S}_n \mid  \Vec{x} >\Vec{0} \right\}. 
\end{align*}
 Any $\Vec{x} \in \mathcal{S}_n^+$ is called a \emph{mixed strategy} (or simply \emph{strategy}) of $A$. 
 A strategy $\Vec{\pi} \in \mathcal{S}_n^+$ is a \emph{Nash equilibrium} of $A$ 
  if it satisfies for any strategy $\Vec{y} \in \mathcal{S}_n^+$ 
   that $\Vec{\pi}^{\top} A \Vec{y} \geq 0$.  
 We say
  a Nash equilibrium $\Vec{\pi}$ is \emph{completely mixed} if it satisfies $\Vec{\pi} \in \mathcal{S}_n^{++}$. 
 We say a game $A$ is \emph{non-redundant} (or completely mixed) 
   if any Nash equilibrium of $A$ is completely mixed. 
 A non-redundant $A$ is completely characterized by Theorems~\ref{thm:kap1} and \ref{thm:kap2} due to Kaplansky\cite{Kaplansky45,Kaplansky95}.

We define the \emph{$\epsilon$-Nash polytope} of $A$ by 
\begin{align}
P_A(\epsilon) 
  = \left\{ \Vec{x} \in \mathcal{S}_n \mid \Vec{x}^{\top}A \geq -\epsilon \Vec{1}^{\top}  \right\}
\label{def:PA}
\end{align}
 for $\epsilon > 0$,  
   where $\Vec{1}=(1,\ldots,1)^{\top}\in \mathbb{R}^n$. 
 We say that $\Vec{x}$ is an \emph{$\epsilon$-Nash equilibrium} of $A$ if $\Vec{x} \in P_A(\epsilon)$. 
 We remark that 
  if $\Vec{x}$ is an $\epsilon$-Nash equilibrium of $A$ then 
  $\Vec{x}^{\top} A \Vec{y} \geq -\epsilon $ holds for any $\Vec{y} \in \mathcal{S}_n^+$. 
 We will use the following fact later. 
\begin{lemma}[cf.\ \cite{maiti2023instance}]\label{eq:quad}
 If $\Vec{x} \in \mathcal{S}_n $ satisfies $|\Vec{x}^{\top}A\Vec{y}| \leq \epsilon$ 
  for any $\Vec{y} \in \mathcal{S}_n^+$ 
 then $\Vec{x} \in P_A(\epsilon)$. 
\end{lemma}
\begin{proof}
 We prove the contraposition. 
 Suppose $\Vec{x} \not\in P_A(\epsilon)$, 
  which means that there exists $i \in [n]$ such that $(\Vec{x}^{\top}A)_i < -\epsilon$. 
 Let $\Vec{y} = \Vec{e}_i$ where $\Vec{e}_i$ ($i=1,\ldots,n$) denote the unit basis, 
  meaning that $\Vec{e}_i = (e_{i,1},\ldots,e_{i,n})$ is given by $e_{i,i}=1$ and $e_{i,j}=0$ for $i \neq j$. 
 Then, $|\Vec{x}^{\top}A\Vec{y}| = |(\Vec{x}^{\top}A)_i| > \epsilon$. 
 We obtain the claim.  
\end{proof}

\subsection{Assumptions}\label{sec:assumption}
\paragraph{On the pay-off matrix $A$}
 We are concerned with 
   a skew-symmetric matrix $A \in [-1,1]^{n \times n}$ 
    for an odd $n \geq 3$. 
 By Theorem~\ref{thm:kap1}, 
  we know $A$ is non-redundant only when 
  $n \geq 3$ is odd and $\rank(A)=n-1$. 
 We also remark on another fact that 
 a non-redundant $A$ must have a solution of   
  $\Vec{x}^{\top} A = \Vec{0}^{\top}$ such that $\sum_{i=1}^n x_i \neq 0$ 
  since its Nash equilibrium $\Vec{\pi}$ satisfies $\Vec{\pi}^{\top} A = \Vec{0}^{\top}$ and $\sum_{i=1}^n \pi_i =1$. 
 Those conditions are summarized as follows, and 
   we basically assume an \emph{unknown} input matrix $A$ satisfies it. 
\begin{condition}\label{cond:A}
 $A \in [-1,1]^{n \times n}$ for an odd $n \geq 3$ 
  is skew-symmetric,  
  satisfies $\rank(A)=n-1$ and 
  has a solution of $\Vec{x}^{\top} A = \Vec{0}^{\top}$ such that $\sum_{i=1}^n x_i \neq 0$. 
\end{condition}

 For Condition~\ref{cond:A}, we remark the following fact, 
  which indicates that it is natural to assume that the rank of a random skew-symmetric matrix is $n-1$ 
  (see Section~\ref{apx:skew} for a proof). 
\begin{proposition}\label{prop:skew}
Let $q$ be a positive integer. 
Let $A = (a_{ij}) \in [-1,1]^{n \times n}$ be a random skew-symmetric matrix, 
 where $q a_{ij}$ for $i<j$ is independently uniformly distributed over integers between $-q$ and $q$, 
 $a_{ij}$ for $i>j$ are given by $a_{ij} = -a_{ji}$, 
 and diagonals are zero. 
Then, $\rank(A)=n-1$ almost surely asymptotic to $q \to \infty$. 
\end{proposition}

\paragraph{On the results of duels}
 We also assume the following condition on the result $X_{ij}$ of a duel our learner receives. 
\begin{condition}\label{cond:X}
 As given an unknown matrix $A=(a_{ij})$, 
  the result $X_{ij}$ of a dual follows 1-sub-Gaussian with mean $a_{ij}$.  
 All results are mutually independent. 
\end{condition}
 Here,  $Z$ is a random variable following \emph{$\sigma^2$-sub-Gaussian} 
  if $\mathbb{P}[|Z -E[Z]| \geq c] \leq 2 \exp(-\frac{c^2}{2\sigma^2})$ holds for all $c \geq 0$ cf.~\cite{Wainwright19}. 
 For instance, the Bernoulli distribution with parameter $p$ is 1/4-sub-Gaussian for any $p \in [0,1]$.
 For another instance,  
  a version of Bernoulli distribution where $Z=1$ with probability $p$ and $Z=-1$ with probability $1-p$ 
    is 1-sub-Gaussian for any $p \in [0,1]$. 
 We will use the following inequality later. 
\begin{theorem}[Hoeffding inequality, cf. \cite{Wainwright19}]\label{thm:Hoeff}
Suppose $Z_i$ for $i=1,\ldots,n$ are iid 1-sub-Gaussian with mean $\mu$. Then, for all $c \geq 0$, 
\begin{align*}
\mathbb{P}\left[ \left| \tfrac{\sum_{i=1}^n Z_i}{n} - \mu \right| \geq c \right] \leq 2\exp\left(-\tfrac{c^2}{2}n \right). 
\end{align*}
\end{theorem}

\subsection{The sample complexity of $3 \times 3$ game}
 We briefly mention the sample complexity of $3 \times 3$ game, 
  which is relatively easy compared with the case of $n \geq 5$. 
 Every skew-symmetric $3 \times 3$ matrix is described by 
\begin{align*}
A=
\begin{pmatrix}
 0 & a &  -b \\
 -a & 0 & c \\
 b & -c & 0 \\
\end{pmatrix}.
\end{align*}
We can observe that $(c,b,a) A = \Vec{0}^{\top}$ holds, thus 
$A$ is non-redundant if\footnote{
  Notice that $\rank(A)=2$ since its eigenvalues are $0$ and $\pm \mathrm{i}\sqrt{a^2+b^2+c^2}$ where $ \mathrm{i}$ is the imaginary unit. 
  } and only if 
$a,b,c$ has the same sign, i.e., $a, b, c > 0$ or  $a, b, c < 0$ (cf.\ Sec.~4 in \cite{Kaplansky45}). 
 Since the probabilities in the Nash equilibrium $\frac{1}{|a+b+c|}(|c|,|b|,|a|)$ directly link to the entries of matrix $A$, 
 we obtain its sample complexity by a standard argument (see Section~\ref{apx:3x3} for proof). 
\begin{theorem}\label{thm:3x3}
The sample complexity is $\Theta(\frac{1}{\Delta^2}\log \frac{1}{\delta})$ where $\Delta = \min\{|a|,|b|,|c|\}$. 
\end{theorem}
 We remark that 
  the Nash equilibrium, that is, a solution of $\Vec{x}^{\top} A = \Vec{0}^{\top}$, 
   links to the entries of $A$ with an affine transformation in the case of $n \geq 5$, 
  which makes our algorithm and analysis described in the following sections difficult.

\section{Identification Whether $A$ Is Non-redundant}\label{sec:upper}
 This section presents 
   an algorithm to identify whether $A$ is non-redundant 
   and proves an upper bound of the sample complexity for $n \geq 5$. 
 The idea behind our algorithm is as follows: 
 When $\|A -\hat{A}\|_{\infty}$ is small enough, 
  then it is natural to expect that a Nash equilibrium $\hat{\Vec{\pi}}$ of $\hat{A}$ 
    approximates the Nash equilibrium $\Vec{\pi}$ of $A$. 
 It is ideal if   
   $\hat{\Vec{\pi}} \in \mathcal{S}_n^{++} \Leftrightarrow \Vec{\pi} \in \mathcal{S}_n^{++}$ holds, 
  but it is not true. 
 To estimate how $\hat{\Vec{\pi}}$ approximates  $\Vec{\pi}$, 
   we use the $\epsilon$-Nash polytope $P_{\hat{A}}(\epsilon)$ of $\hat{A}$.

\subsection{Understanding the $\epsilon$-Nash polytope  --- as a preliminary step}
 To explain an intuition of the parameters appearing in our algorithm and theorem, 
 this section establishes Lemmas ~\ref{lem:xA_i} and \ref{lem:vertex} 
  respectively about the solution of  $\Vec{x}^{\top}A=\Vec{0}^{\top}$ and 
  the $\epsilon$-Nash polytope $P_A(\epsilon)$, as a preliminary step. 

 Let $A_j$ for $j=1,\ldots,n$ be 
 the matrix formed by replacing the $j$-th column of $A$ with the column vector~$\Vec{1}$. 
For instance, 
\begin{align*}
A_1 = 
  \begin{pmatrix}
   1 & a_{12} &\cdots & a_{1n} \\
   \vdots& \vdots&  & \vdots \\
   1 & a_{n2} & \cdots & a_{nn} 
  \end{pmatrix}.
\end{align*}
 Firstly, we remark the following fact which implies that 
  Condition~\ref{cond:A} ensures $A_j^{-1}$ for $j=1,\ldots,n$ 
  (see Section~\ref{apx:upper} for a proof of Lemma~\ref{lem:rankAj}). 
\begin{lemma}\label{lem:rankAj}
Suppose $\rank(A) = n-1$. 
$A_j$ is non-singular  for all $j \in \{1,\ldots,n\}$
if and only if $\exists \Vec{x} \neq \Vec{0}$ such that $\Vec{x}^{\top}A=\Vec{0}^{\top}$ and $\sum_{i=1}^n x_i \neq 0$. 
\end{lemma}

 Next, the following lemma gives the solution of $\Vec{x}^{\top}A=\Vec{0}^{\top}$ using $A_j^{-1}$.
\begin{lemma}\label{lem:xA_i}
 Suppose $A \in [-1,1]^{n \times n}$ satisfies Condition~\ref{cond:A}. 
 Let $\Vec{\pi} \in \mathcal{S}_n$ satisfy $\Vec{\pi}^{\top}A = \Vec{0}^{\top}$. 
 Then, 
\begin{align}
 \Vec{\pi}^{\top} = \Vec{e}_j^{\top}A_j^{-1} 
\label{eq:x1}
\end{align}
  holds for each $j = 1,\ldots,n$.
\end{lemma}
\begin{proof}
 $\rank(A) = n-1$ implies that $\dim(\ker(A))=1$. 
 Let $\Vec{c} \in \ker(A) \setminus \{\Vec{0}\}$, i.e., 
   $\Vec{c}^{\top}A=\Vec{0}^{\top}$ and $\Vec{c} \neq \Vec{0}$. 
 Let $\Vec{\pi} = \frac{1}{\sum_{i=1}^n c_i}\Vec{c}$, 
   where $\sum_{i=1}^n c_i \neq 0$ by Condition~\ref{cond:A}. 
 Then, $\Vec{\pi}$ is the (unique) solution of $\Vec{x}^{\top}A=0$ and $x_1 + \cdots + x_n = 1$. 
 Now, it is not difficult to observe that 
 $ \Vec{\pi}^{\top}A_j = \Vec{e}_j^{\top}$.  
 Recall $A_j$ is non-singular by Lemma~\ref{lem:rankAj}. 
\end{proof}

The following lemma presents the vertices of $P_A(\epsilon)$. 
\begin{lemma}\label{lem:vertex}
 Suppose $A \in [-1,1]^{n \times n}$ satisfies Condition~\ref{cond:A}. 
Let
\begin{align}
 \Vec{v}_j^{\top} 
  &= \Vec{\pi}^{\top}-\epsilon(\Vec{1}^{\top}A_j^{-1}-\Vec{\pi}^{\top}) 
\label{eq:v1}
\end{align}
for any $j=1,\ldots,n$ where $\Vec{\pi}$ is given by \eqref{eq:x1}. Then, 
\begin{align*}
P_A(\epsilon) 
  = \conv\{ \Vec{v}_1, \ldots, \Vec{v}_n \}
\end{align*}
 where $\conv\ \! S$ denotes the convex hull of $S \subseteq \mathbb{R}^n$ (see, e.g., \cite{Matousek}). 
\end{lemma}
\begin{proof}
Recall \eqref{def:PA}, that is 
$P_A(\epsilon)$ is described by the following $n$ inequalities and one equality: 
\begin{align*}
 & x_1 a_{1i}+x_2 a_{2i}+\cdots+x_n a_{ni} \geq -\epsilon \ \mbox{for $i \in [n]$, and }\\
 & x_1+x_2+\cdots+x_n=1. 
\end{align*}
By a standard argument of linear algebra, we can see for each $j=1,\ldots,n$  that 
\begin{align}
\begin{array}{ll}
 & x_1 a_{1i}+x_2 a_{2i}+\cdots+x_n a_{ni}= -\epsilon \ \mbox{for $i \in [n] \setminus \{j\}$, and }\\[0.5ex]
 & x_1+x_2+\cdots+x_n=1
\end{array}
\label{form:P}
\end{align}
gives a vertex of $P_A(\epsilon)$. 
\eqref{form:P} is described by 
 $\Vec{x}^{\top}A_j 
  = -\epsilon \sum_{i \neq j} \Vec{e}_i^{\top} + \Vec{e}_j^{\top}
  = -\epsilon(\Vec{1}-\Vec{e}_j)^{\top} + \Vec{e}_j^{\top}$. 
Since $A_j$ is non-singular by Lemma~\ref{lem:rankAj}, the solution is given by 
  $\Vec{x}^{\top} 
   = (-\epsilon(\Vec{1}-\Vec{e}_j)^{\top} + \Vec{e}_j^{\top})A_j^{-1}
   = -\epsilon(\Vec{1}-\Vec{e}_j)^{\top}A_j^{-1} + \Vec{\pi}^{\top}
   = -\epsilon(\Vec{1}^{\top}A_j^{-1}-\Vec{\pi}^{\top}) + \Vec{\pi}^{\top}$ 
 where we used $\Vec{e}_j^{\top} A_j^{-1} = \Vec{\pi}^{\top}$ by Lemma~\ref{lem:xA_i}. 
We obtain the claim. 
\end{proof}

\subsection{Algorithm and theorem}

\begin{algorithm}[t]
\caption{Identify if $A$ is non-redundant}\label{alg:cm}
\nl $\hat{A} = (\hat{a}_{ij})_{n \times n}$, $Z_{ij} \gets 0$; \\
\nl $T\gets \lceil \frac{2U^2}{\alpha^2} \log \frac{2n^2}{\delta} \rceil + 1$;\\
\nl \For{$t = 1,2,\ldots,T $}{
\nl  \For{$\{i,j\} \in \binom{[n]}{2}$}{
  \nl   receive the result $X_{ij}$ of a duel between $i$ and $j$ according to unknown $A$; \\
  \nl   $Z_{ij} \gets Z_{ij} + X_{ij}$, $\hat{a}_{ij} \gets \frac{Z_{ij}}{t}$, $\hat{a}_{ji} \gets -\hat{a}_{ij}$;
  }
\nl  Compute $\hat{\Vec{\pi}} = \Vec{e}_1^{\top} \hat{A}_1^{-1}$;\\
\nl  Compute $ \varphi(\hat{A})= 
  \underset{j \in [n]}{\max}\ \! \underset{i \in [n]}{\max} 
    \left| \left( \Vec{1}^{\top} \hat{A}_j^{-1} -\hat{\Vec{\pi}} \right)_i \right|$;\\
\nl  \If{
    $t > \frac{2\varphi(\hat{A})^2}{\hat{\pi}_{\min}^2}\log \frac{2n^2}{\delta}$ and 
           $\hat{\pi}_{\min} >0$}{
\nl            Conclude ``$A$ is non-redundant'' and terminate Algorithm~\ref{alg:cm};  \hfill(a) 
            }
\nl  \If{$t > \frac{2\varphi(\hat{A})^2}{\alpha^2}\log \frac{2n^2 }{\delta}$}{
\nl    $\epsilon \gets 
     \frac{\alpha}{\varphi(\hat{A})}$;\\
\nl    Compute $\hat{\Vec{v}}_j^{\top} = \hat{\Vec{\pi}}^{\top}-\epsilon(\Vec{1}^{\top}\hat{A}_j^{-1} - \hat{\Vec{\pi}}^{\top})$ 
       for all $j=1,2,\ldots,n$;\\
\nl    \If{$\hat{v}_{j,\min} < \alpha$ for all $j=1,2,\ldots,n$}{
\nl      Conclude ``$A$ is $\alpha$-redundant'' and terminate Algorithm~\ref{alg:cm};  \hfill(b) 
     }
   }
  }	
\nl    $\epsilon \gets \frac{\alpha}{U}$;\\
\nl    Compute $\hat{\Vec{v}}_j^{\top} = \hat{\Vec{\pi}}^{\top}-\epsilon(\Vec{1}^{\top}\hat{A}_j^{-1} - \hat{\Vec{\pi}}^{\top})$ 
       for all $j=1,2,\ldots,n$;\\
\nl  \If{$\hat{v}_{j,\min} > 0$ for all $j=1,2,\ldots,n$}
{\nl Conclude ``$A$ is non-redundant'';   \hfill(c) }
\nl \Else{\nl Conclude ``$A$ is $\alpha$-redundant'';   \hfill(d) }
\end{algorithm}

 For a pay-off matrix $A$, 
  let $\Vec{\pi} = (\pi_1\ldots,\pi_n) \in \mathcal{S}_n$ satisfy $\Vec{\pi}^{\top} A =\Vec{0}^{\top}$.  
 Let $\pi_{\min} = \min_{i \in [n]} \pi_i$. 
 Let $P(\epsilon) = P_A(\epsilon)$ for $\epsilon > 0$, for convenience.  
 Similarly, for an estimated pay-off matrix $\hat{A}$, 
  let $\hat{\Vec{\pi}} \in \mathcal{S}_n$ satisfy $\hat{\Vec{\pi}}^{\top} \hat{A} =\Vec{0}^{\top}$, 
  let $\hat{\pi}_{\min} = \min_{i \in [n]} \hat{\pi}_i$, and
  let $\hat{P}(\epsilon) = P_{\hat{A}}(\epsilon)$. 
Recalling \eqref{eq:v1}, 
  we define 
\begin{align}
 \varphi(A)= 
  \underset{j \in [n]}{\max}\ \! \underset{i \in [n]}{\max} 
    \left| \left( \Vec{1}^{\top} A_j^{-1} -\Vec{\pi} \right)_i \right|, 
\label{def:alpha}
\end{align}    
 where we remark 
$ \varphi(\hat{A})= 
  \underset{j \in [n]}{\max}\ \! \underset{i \in [n]}{\max} 
    \left| \left( \Vec{1}^{\top} \hat{A}_j^{-1} -\hat{\Vec{\pi}} \right)_i \right|$ just in case.  
 Intuitively,
  $\varphi(A)$ gets large if $\det(A_j)$ is close to $0$ for some $j$ (see Section~\ref{sec:low} for such an example).

  Let 
\begin{align}
  \mathcal{S}_n^{\alpha} = \{ \Vec{x}\in \mathcal{S}_n \mid \Vec{x} \geq \alpha \Vec{1}\}
\label{def:sntheta}
\end{align}  
   for $\alpha \in (0,\frac{1}{n}]$.  
 We say $A$ is \emph{$\alpha$-redundant} 
   if $A$ does not have any Nash equilibrium $\Vec{\pi}$
  satisfying $\Vec{\pi} \in \mathcal{S}_n^{\alpha}$. 
 Now, we present Algorithm~\ref{alg:cm}. 
 Roughly speaking, 
  the algorithm estimates $A$ at $\hat{A}$ by repeating duels. 
 If the number of iterations gets large enough to assume $|A_{ij} - \hat{A}_{ij}|$ is sufficiently small 
  then the algorithm decides whether $A$ is non-redundant or $\alpha$-redundant 
  by checking the estimated Nash equilibrium $\hat{\pi}$ or the $\epsilon$-Nash region $\hat{P}(\epsilon)$. 
 It terminates in 
   $\Order\!\left(\frac{\varphi(A)^2}{\max\{\alpha^2,\pi_{\min}^2\}}\log \frac{n}{\delta}\right)$ rounds 
   if we know in advance that $\varphi(A)$ is not very big. 
 Note that 
  $\hat{v}_{j,\min} = \min_{j \in [n]} \hat{v}_{j,i}$ 
   for $\hat{\Vec{v}}_j = (\hat{v}_{j,1},\ldots,\hat{v}_{j,n})$
 in Algorithm~\ref{alg:cm}. 

\begin{theorem}\label{thm:upper}
 Suppose $A \in [-1,1]^{n \times n}$ satisfies Condition~\ref{cond:A}, and 
 suppose that we can assume that $\varphi(A) \leq U$. 
 Then, Algorithm~\ref{alg:cm} correctly concludes about the nonredundancy of $A$ with probability at least $1-\delta$. 
 The sample complexity of Algorithm~\ref{alg:cm} is 
  $\Order\!\left(\frac{U^2}{\max\{\alpha^2,\pi_{\min}^2\}}\log \frac{n}{\delta}\right)$. 
\end{theorem}

 The sample complexity is trivial from Algorithm~\ref{alg:cm}. 
 Thus, the heart of the proof of Theorem~\ref{thm:upper} is the correctness of the conclusions (a)--(d). 
 For the purpose, 
  the following Lemmas~\ref{lem:P_Phat1}--\ref{lem:P_epsilon} are the key. 
\begin{lemma}\label{lem:P_Phat1}
 Let $B=(b_{ij})$ and $C=(c_{ij})$ be $n \times n$ skew-symmetric matrices 
  satisfying 
  $\max_{i,j}|b_{ij} -c_{ij} | \leq \epsilon$.  
If $P_B(\epsilon') \subseteq \mathcal{S}_n^+$ 
 then $P_B(\epsilon') \subseteq P_C(\epsilon+\epsilon')$. 
\end{lemma}
\begin{proof}
 We prove that any  $\Vec{x} \in P_B(\epsilon')$ satisfies  $\Vec{x} \in P_C(\epsilon+\epsilon')$. 
 Suppose $\Vec{x} \in P_B(\epsilon)$. 
 Then, 
\begin{align*}
 \Vec{x}^{\top} C = \Vec{x}^{\top} B + \Vec{x}^{\top}(C-B) 
\end{align*}
 holds. 
 The hypothesis $\Vec{x} \in P_B(\epsilon)$ implies 
  $\Vec{x}^{\top} B \geq -\epsilon \Vec{1}$. 
 The hypothesis $\max_{i,j}|B_{ij} -C_{ij} | < \epsilon$ implies 
  for any $\Vec{y} \in \mathcal{S}_n^+$ that
\begin{align*}
 |\Vec{x}^{\top}(C-B)\Vec{y}| 
  &= \left| \sum_{i=1}^n x_i  \sum_{j=1}^n (C-B)_{ij} y_j \right| 
  \leq \sum_{i=1}^n |x_i| \sum_{j=1}^n \left| (C-B)_{ij} \right|  |y_j| \\
  &= \sum_{i=1}^n x_i \sum_{j=1}^n \left| b_{ij}-c_{ij} \right|  y_j 
  \leq  \epsilon \sum_{i=1}^n x_i \sum_{j=1}^n  y_j 
  =  \epsilon 
\end{align*} 
  where we used 
    $\Vec{x} \geq \Vec{0}$, 
    $\Vec{y} \geq \Vec{0}$, 
    $\max_{i,j}|b_{ij} - c_{ij}| \leq \epsilon$ and 
    $\sum_{i=1}^n x_i = \sum_{i=1}^n y_i = 1$. 
 This implies $\Vec{x}^{\top}(C-B) \geq -\epsilon \Vec{1}$ by Lemma~\ref{eq:quad}. 
 Now the claim is easy. 
\end{proof}

The following Lemma~\ref{lem:epsilonNE} is a special case of Lemma~\ref{lem:P_Phat1}. 
\begin{lemma}\label{lem:epsilonNE}
 Let $B=(b_{ij})$ and $C=(c_{ij})$ be $n \times n$ skew-symmetric matrices 
  satisfying 
  $\max_{i,j}|b_{ij} -c_{ij} | \leq \epsilon$.  
 Let $\Vec{x} \in \mathcal{S}_n^+$ satisfy $\Vec{x}^{\top} B \geq \Vec{0} $. 
 Then, 
  $\Vec{x} \in P_C(\epsilon)$.  
\end{lemma}

\begin{lemma}\label{lem:P_epsilon}
Let $B \in [-1,1]^{n \times n}$ satisfy Condition~\ref{cond:A}. 
Let $\Vec{x}^{\top}B = \Vec{0}^{\top}$. 
Suppose $\Vec{x} \in \mathcal{S}_n^{++}$. 
If a nonnegative $\epsilon$ satisfies 
 \begin{align}
  \epsilon < 
  \frac{x_{\min}}{\varphi(B)}
\label{eq:epsilon}
\end{align}
 then $P_B(\epsilon) \subseteq \mathcal{S}_n^{++}$.
\end{lemma}
\begin{proof}
By Lemma~\ref{lem:vertex}, we know
 $\Vec{v}_j^{\top}   = \Vec{x}^{\top} -\epsilon (\Vec{1}^{\top} B_j^{-1} -\Vec{x}^{\top})$ 
is a vertex of $P_B(\epsilon)$. 
Let $\Vec{v}_j = (v_{j,1},\ldots,v_{j,n})^{\top} $ for convenience. 
Then, 
\begin{align*}
v_{j,i}
 &\geq x_i -\epsilon \left(\Vec{1}^{\top} B_j^{-1} -\Vec{x}^{\top}\right)_i \\
 &\geq x_i-\epsilon \left| \left(\Vec{1}^{\top} B_j^{-1} -\Vec{x}^{\top} \right)_i \right| \\
 &> x_i-x_{\min} & (\mbox{by $\epsilon \geq 0$ and \eqref{eq:epsilon}})\\
 &\geq 0
\end{align*}
holds for any $i \in [n]$ and $j \in [n]$. 
Now the claim is easy. 
\end{proof}

Next, we prove Lemmas~\ref{lem:1}--\ref{lem:3} that respectively correspond to conclusions (a), (b) and (d). 
\begin{lemma}\label{lem:1}
 Suppose $\hat{\Vec{\pi}} \in \mathcal{S}_n^{++}$. 
 If $\underset{i,j}{\max}\ \!|\hat{a}_{ij} - a_{ij}| < \epsilon_1$ 
 where 
  $\epsilon_1 = \hat{\pi}_{\min}/\varphi(\hat{A})$
 then $\pi \in \mathcal{S}_n^{++}$. 
\end{lemma}
\begin{proof}
By Lemma~\ref{lem:epsilonNE}, $\Vec{\pi} \in \hat{P}(\epsilon_1)$. 
By Lemma~\ref{lem:P_epsilon}, 
$\hat{P}(\epsilon_1) \subseteq \mathcal{S}_n^{++}$. 
\end{proof}

\begin{lemma}\label{lem:2}
 Suppose $\underset{i,j}{\max}\ \!|\hat{a}_{ij} - a_{ij}| < \epsilon_2$  
 where 
 $\epsilon_2 = \alpha/\varphi(\hat{A})$ for $ 0< \alpha < \frac{1}{n}$. 
 If $\hat{P}(\epsilon_2) \cap \mathcal{S}_n^{\alpha} = \emptyset$
   then $\Vec{\pi} \not\in \mathcal{S}_n^{\alpha}$. 
\end{lemma}
\begin{proof}
$\Vec{\pi} \in \hat{P}(\epsilon_2)$ by Lemma~\ref{lem:epsilonNE}. 
 If $\hat{P}(\epsilon_2) \cap \mathcal{S}_n^{\alpha} = \emptyset$
   then $\pi_{\min} < \alpha$. 
\end{proof}

\begin{lemma}\label{lem:3}
 Suppose $\underset{i,j}{\max}\ \!|\hat{a}_{ij} - a_{ij}| < \epsilon_3$ 
 where 
  $\epsilon_3 =  \alpha/\varphi(A)$.  
 If $\hat{P}(2\epsilon_3) \not\subseteq \mathcal{S}_n^{++}$
   then $\Vec{\pi} \not\in \mathcal{S}_n^{\alpha}$. 
\end{lemma}
\begin{proof}
We prove contraposition: 
 If $\Vec{\pi} \in \mathcal{S}_n^{\alpha}$ then $\hat{P}(\epsilon_3) \subseteq \mathcal{S}_n^{++}$. 
If $\Vec{\pi} \in \mathcal{S}_n^{\alpha}$ then $P(2\epsilon_3) \subseteq \mathcal{S}_n^{++}$ 
by Lemma~\ref{lem:P_epsilon}.
By Lemma~\ref{lem:P_Phat1}, 
$\hat{P}(\epsilon_3) \subseteq P(2\epsilon_3)$. 
\end{proof}


 We use the following lemma, easily derived from Hoeffding's inequality. 
\begin{lemma}\label{lem:PACtime}
Let $\hat{A}^{(t)} = (\hat{a}^{(t)}_{ij})$ denote the estimated pay-off matrix at the end of the $t$-th iteration of Algorithm~\ref{alg:cm}. 
If $t \geq \frac{2}{\epsilon^2}\log \frac{2n^2}{\delta}$ then 
\begin{align*} 
 \mathbb{P}\left[\left|\hat{a}^{(t)}_{ij} - a_{ij} \right| < \epsilon \mbox { for all $\{i,j\} \in \binom{[n]}{2}$}\right] >1 - \delta. 
\end{align*}
\end{lemma}
\begin{proof}
By Hoeffding's inequality (Theorem~\ref{thm:Hoeff}), 
\begin{align*} 
 \mathbb{P}[ |\hat{a}^{(t)}_{ij} - a_{ij} | \geq \epsilon] 
 \leq 2\exp\left(-\tfrac{\epsilon^2}{2}t\right)
 \leq 2\exp\left(-\tfrac{\epsilon^2}{2} \left(\tfrac{2}{\epsilon^2}\log \tfrac{2n^2}{\delta}\right) \right)
 = \tfrac{\delta}{n^2}
\end{align*}
 for any $i,j$. By union bound, 
\begin{align*} 
 \mathbb{P}\left[\left|\hat{a}^{(t)}_{ij} - a_{ij} \right| < \epsilon \mbox { for all $\{i,j\} \in \binom{[n]}{2}$}\right] 
 &\geq 1 - \sum_{i,j} \left(1- \mathbb{P}[ |\hat{a}^{(t)}_{ij} - a_{ij} | \geq \epsilon] \right) \\
 &\textstyle \geq 1 - \binom{n}{2} \frac{\delta}{n^2} 
 \geq 1- \delta 
\end{align*}
 and we obtain the claim. 
\end{proof}

\begin{proof}[Proof of Theorem~\ref{thm:upper}]
 If  $t > \frac{2\varphi(\hat{A})^2}{\hat{\pi}_{\min}^2}\log \frac{2n^2}{\delta}$ 
  then $\max_{i,j}|a_{ij}-\hat{a}_{ij}| < \frac{\hat{\pi}_{\min}}{\varphi(\hat{A})}$ by Lemma~\ref{lem:PACtime}. 
 Thus, the conclusion (a) is correct by Lemma~\ref{lem:1}. 
 If $t > \frac{2\varphi(\hat{A})^2}{\alpha^2}\log \frac{2n^2 }{\delta}$ 
  then $\max_{i,j}|a_{ij}-\hat{a}_{ij}| < \frac{\alpha}{\varphi(\hat{A})}$ by Lemma~\ref{lem:PACtime}. 
 Thus, the conclusion (b) is correct by Lemma~\ref{lem:2}. 
 The conclusion (c) is trivial from lemma~\ref{lem:epsilonNE}\footnote{This case could be detected much earlier as the conclusion (a). }. 
 If $t > \frac{2U^2}{\alpha^2} \log \frac{2n^2}{\delta} $
  then $\max_{i,j}|a_{ij}-\hat{a}_{ij}| < \frac{\alpha}{\varphi(A)}$ by Lemma~\ref{lem:PACtime} since $\varphi(A) \leq U$. 
 Thus, the conclusion (d) is correct by Lemma~\ref{lem:3}. 
 The sample complexity is trivial. 
\end{proof}


\section{Lower Bound of the Sample Complexity}\label{sec:low}
 Concerning the lower bounds of the sample complexity of our problem for $n \geq 5$, 
  we can prove the following two theorems, 
   where Theorem~\ref{thm:low-varphi} is supported by computer-aided symbolic calculations. 
 See Section~\ref{apx:low} for proofs. 
\begin{theorem}\label{thm:ETlow}
 Let $\alpha$ and $\delta$ be fixed parameters respectively satisfying $0 < \alpha \ll 1$ and $0 < \delta \ll 1$. 
 Let $\tau$ denote the running time of 
   an arbitrary $(\alpha,\delta)$-PAC algorithm 
   that identifies whether an arbitrarily given $A$ is non-redundant.  
 Then,  the expected running time satisfies 
  $\mathbb{E}{[\tau]}\geq \frac{1}{2\alpha^2}\log{\frac{5}{12\delta}}$. 
\end{theorem}
\begin{theorem}\label{thm:low-varphi}
 Let $\delta$ be a fixed parameter satisfying $0 < \delta \ll 1$. 
 Let $\tau$ denote the running time of 
   an arbitrary $(\alpha,\delta)$-PAC algorithm 
   that identifies whether an arbitrarily given $A$ is non-redundant.  
 Then, 
  $\mathbb{E}{[\tau]} = \Omega\!\left(\varphi(A)^2 \log\frac{1}{\delta}\right)$
   for each $n=5,7,\ldots,19$. 
\end{theorem}
 As a consequence of them, 
 we obtain a lower bound 
  $\mathbb{E}{[\tau]} = \Omega\!\left(\max\{\frac{1}{\alpha^2},\varphi(A)^2\} \log\frac{1}{\delta}\right)$
    for each $n=5,7,\ldots,19$.

\section{Concluding Remarks}\label{sec:conclusion}
 Focusing on the nontransitive relation, 
   this paper introduced a dueling bandit problem to identify the non-redundancy of moves. 
 We gave an algorithm with 
   $\Order\!\left(\frac{\varphi(A)^2}{\max\{\alpha^2,\pi_{\min}^2\}}\log \frac{n}{\delta}\right)$ samples. 
 We also gave lower bounds of the sample complexity of the problem 
   $\Omega\!\left(\frac{1}{\alpha^2}\log{\frac{1}{\delta}}\right)$ and 
   $\Omega\!\left(\varphi(A)^2 \log\frac{1}{\delta}\right)$. 
 Filling the gap between upper and lower bounds is a future work. 

 Our algorithm and analysis may feel somehow complicated. 
 A better understanding of the gap between $\varphi(A)$ and $\varphi(\hat{A})$ could 
  provide a simpler algorithm and proof. 
 This paper employed sequential sampling of duels following the work of \citet{maiti2023instance}. 
 Adaptive sampling of duels is a future work. 
 This paper was concerned with identifying whether all moves are indispensable. 
 Finding all indispensable moves in a game is another work. 

\bibliography{cmneref}
\bibliographystyle{abbrvnat}
\newpage
\appendix
\section{Supplemental Proofs of Section~\ref{sec:prelim}}\label{apx:prelim}
\subsection{Proof of Proposition~\ref{prop:skew}}\label{apx:skew}
\begin{proof}[Proof of Proposition~\ref{prop:skew}]
 Let $F_q$ be a finite field of size $q \geq 2$, and 
 consider a uniformly random skew-symmetric matrix $M_n \in F_q^{n \times n}$. 
 Let $N(q,n)$ denote the total number of such matrices of size $n$, and let 
 $S(q,n,2r)$ denote the number of matrices of rank $2r$, 
 noting that the rank of skew-symmetric matrix over any field is always even. 
 It is easy to know that 
\begin{align}
 N(q,n) = q^{\binom{n}{2}}
\end{align}
and by Theorem 3 of Carlitz \cite{Carlitz54} (cf \cite{FG15}), the number of matrices of rank $2r$ is given by 
\begin{align}
S(q,n,2r) = q^{r(r-1)} \frac{\prod_{i=0}^{2r-1}(q^{n-i}-1)}{\prod_{i=1}^{r}(q^{2i}-1)}.
\end{align}
In particular, when $n$ is odd, the maximal possible rank is $n-1$, and asymptotically (as $q \to \infty$), 
the probability that a random skew-symmetric matrix $M_n$ over $F_q$ has rank $n-1$ satisfies 
\begin{align*}
P(\rank(M_n) = n-1) 
 &= \frac{S(q,n,n-1)}{N(q,n)} \\
 &= \frac{q^{\frac{n-1}{2}(\frac{n-1}{2}-1)} \frac{\prod_{i=1}^{n-2}(q^{n-i}-1)}{\prod_{i=1}^{\frac{n-1}{2}} (q^{2i}-1)}}{q^{\binom{n}{2}}} \\
 &= \frac{q^{\frac{(n-1)(n-3)}{4}} \frac{(q^n-1)(q^{n-1}-1)\cdots(q^2-1)}{(q^2-1)(q^4-1)\cdots(q^{n-1}-1)}}{q^{\frac{n(n-1)}{2}}} \\
 &= \frac{q^{\frac{(n-1)(n-3)}{4}} (q^3-1)(q^5-1)\cdots(q^n-1)}{q^{\frac{n(n-1)}{2}}} \\
 &= \frac{q^{\frac{(n-1)(n-3)}{4}} q^3q^5\cdots q^n (1-\frac{1}{q^3}) (1-\frac{1}{q^5})\cdots (1-\frac{1}{q^n})}{q^{\frac{n(n-1)}{2}}} \\
 &\geq \frac{q^{\frac{(n-1)(n-3)}{4}} q^3q^5\cdots q^n (1-\frac{1}{q^3})^{\frac{n-1}{2}} }{q^{\frac{n(n-1)}{2}}} \\
 &= \frac{q^{\frac{(n-1)(n-3)}{4} + \frac{(n-1)(n+3)}{4}} }{q^{\frac{n(n-1)}{2}}} (1-\frac{1}{q^3})^{\frac{n-1}{2}}\\
 &= \left(1-\frac{1}{q^3}\right)^{\frac{n-1}{2}}\\
 &\geq 1-\frac{n-1}{2q^3}
\end{align*}
\end{proof}

\subsection{Proof of Theorem~\ref{thm:3x3}}\label{apx:3x3}
\begin{algorithm}[t]
\caption{Identify if $A \in [-1,1]$ is non-redundant}\label{alg:3x3}
\nl $\hat{A} = (\hat{a}_{ij})_{3 \times 3}$, $Z_{ij} \gets 0$; \\
\nl \For{$t = 1,2,\ldots,T $}{
\nl  \For{$\{i,j\} \in \binom{[3]}{2}$}{
  \nl   get the result $X_{ij}$ of a duel between $i$ and $j$ according to unknown $A$; \\
  \nl   $Z_{ij} \gets Z_{ij} + X_{ij}$, $\hat{a}_{ij} \gets \frac{Z_{ij}}{t}$, $\hat{a}_{ji} \gets -\hat{a}_{ij}$;
  }
\nl Set $\hat{\Delta} = \min \{ |\hat{a}_{12}|, |\hat{a}_{23}|, |\hat{a}_{31}|\} $;\\
\nl  \If{$t > \frac{18}{\hat{\Delta}^2}\log \frac{2}{\delta}$}{
\nl     \If{$\hat{a}_{12}$, $\hat{a}_{23}$ and $\hat{a}_{31}$ has the same sign}{
\nl          Conclude ``$A$ is non-redundant'';}
\nl     \Else{
\nl            Conclude ``$A$ is redundant'';  
             }
\nl     terminate Algorithm~\ref{alg:cm}
         }
   }	
\end{algorithm}

 This section proves Theorem~\ref{thm:3x3}, which immediately follows from 
   an upper bound given by Lemma~\ref{lem:up3x3} and 
   a lower bound given by Lemma~\ref{lem:low-3x3} appearing below. 
\subsubsection{Upper bound}
 Firstly, we prove an upper bound. 
\begin{lemma}\label{lem:up3x3}
 Algorithm~\ref{alg:3x3} correctly identifies whether $A$ is non-redundant 
 in $\Order(\frac{1}{\Delta^2}\log \frac{1}{\delta})$ rounds 
  with probability at least $1-\delta$. 
\end{lemma}
\begin{proof}
Let $\hat{A}^{(t)} = (\hat{a}^{(t)}_{ij})$ denote the estimated pay-off matrix at the end of the $t$-th iteration of Algorithm~\ref{alg:3x3}. 
Suppose $t \geq \frac{8}{\Delta^2}\log \frac{6}{\delta}$ where $\Delta = \min\{|a_{12}|,|a_{23}|,|a_{31}|\}$. 
By Hoeffding's inequality (Theorem~\ref{thm:Hoeff}), 
\begin{align*} 
 \mathbb{P}[ |\hat{a}^{(t)}_{ij} - a_{ij} | \geq \tfrac{\Delta}{2}] 
 \leq 2\exp\left(-\tfrac{(\frac{\Delta}{2})^2}{2}t\right)
 \leq 2\exp\left(-\tfrac{\Delta^2}{8} \left(\tfrac{8}{\Delta^2}\log \tfrac{6}{\delta}\right) \right)
 = \frac{\delta}{3}
\end{align*}
 for any $i,j$. By union bound, 
\begin{align*} 
 \mathbb{P}\left[\left|\hat{a}^{(t)}_{ij} - a_{ij} \right| < \tfrac{\Delta}{2} \mbox { for all $\{i,j\} \in \binom{[3]}{2}$}\right] 
 &\geq 1 - \sum_{i,j} \left(1- \mathbb{P}[ |\hat{a}^{(t)}_{ij} - a_{ij} | \geq \tfrac{\Delta}{2}] \right) \\
 &\geq 1 - \binom{3}{2} \frac{\delta}{3} 
 \geq 1- \delta 
\end{align*}
 and we obtain 
\begin{align*} 
 \mathbb{P}\left[\left|\hat{a}^{(t)}_{ij} - a_{ij} \right| < \tfrac{\Delta}{2} \mbox { for all $\{i,j\} \in \binom{[3]}{2}$}\right] >1 - \delta. 
\end{align*}

If $ [|\hat{a}^{(t)}_{ij} - a_{ij} | < \tfrac{\Delta}{2} \mbox { for all $\{i,j\} \in \binom{[3]}{2}$}] $ then 
  $\hat{\Delta} \leq \Delta + \frac{\Delta}{2}$ holds,
  which implies 
 $\frac{18}{\hat{\Delta}^2}\log \frac{2}{\delta} \geq \frac{8}{\Delta^2}\log \frac{6}{\delta}$. 
Thus, Algorithm~\ref{alg:3x3} concludes correctly with probability at least $1-\delta$. 
\end{proof}

\subsubsection{Lower bound}
 We use the following Bretagnolle–Huber inequality (instead of Pinsker's inequality) for a lower bound. 
\begin{theorem}[Bretagnolle–Huber inequality~\cite{Canonne23}]\label{thm:BHineq}
 Let $\nu$ and $\nu'$ be two probability distributions over~$\Omega$. 
 Let $\dtv(\nu,\nu') = \sup_{A\subseteq \Omega}\{ |\nu(A) - \nu(A)|\}$. Then, 
\begin{align*} 
 \dtv(\nu,\nu') \leq 1-\frac{1}{2} \exp(-D_{KL} (\nu \parallel \nu')).
\end{align*}
\end{theorem}

The following lemmas gives the Kullback-Leibler divergence for the normal distributions. 
\begin{theorem}[cf.~\cite{Soch}]\label{prop-normal}
 Let $P = N(\mu_1,\sigma_1)$ and $Q = N(\mu_2,\sigma_2)$, 
 then the Kullback-Leibler divergence $\dkl(P,Q)$ of $P$ from $Q$ satisfies 
\begin{align}
\dkl(P,Q) = 
  \frac{1}{2}\left(\frac{(\mu_2-\mu_1)^2}{\sigma_2^2} + \frac{\sigma_1^2}{\sigma_2^2}
  -\ln \frac{\sigma_1^2}{\sigma_2^2} -1\right). 
\end{align}
\end{theorem}

 Now, we prove a lower bound.
\begin{lemma}\label{lem:low-3x3}
 Let $\delta$ be fixed parameters respectively satisfying $0 < \delta \ll 1$. 
 Let $\tau$ denote the running time of 
   an arbitrary $\delta$-PAC algorithm 
   that identifies whether a $3 \times 3$ skew-symmetric matrix $A$ is non-redundant.  
 Then, 
  $\tau\geq \frac{1}{2\Delta^2}\log{\frac{1}{\delta}}$ where $\Delta = \min\{ |a_{12}|,|a_{23}|,|a_{31}| \}$. 
\end{lemma}

\begin{proof}
Let 
$Q^+=
\begin{pmatrix}
 0 & a &  -b \\
 -a & 0 & c \\
 b & -c & 0 \\
\end{pmatrix}
$ and 
$Q^-=
\begin{pmatrix}
 0 & -a &  -b \\
 a & 0 & c \\
 b & -c & 0 \\
\end{pmatrix}
$ 
where we may assume $a,b,c,>0$ and $\min\{a,b,c\} = a$ without loss of generality. 
 Note that $Q^+$ is non-redundant while $Q^-$ is redundant. 
 We are concerned with the following decision problem under the dueling bandit setting: 
   given an unknown matrix $A=(a_{ij})$ such that $A \in \{Q^+,Q^-\}$ and 
   decide whether $A$ is redundant or not 
   by observing the results of duels. 
 Here,  the result $X_{ij}$ ($i<j$) of duels is deterministic unless $(i,j) = (1,2)$
   so that $X_{23} = c$ and $X_{31} = b$  at any time, and  
 the result $X_{12}$ follows the normal distribution $\mathcal{N}(a, 1)$ with mean $a$ and variance $1$.  
 For convenience, 
  let $\nu^+ = \mathcal{N}(a, 1)$ and $\nu^- = \mathcal{N}(-a, 1)$,  
  thus $X_{12}$ follows $\nu^+$ if $A=Q^+$, otherwise $\nu^-$. 
 We note that the KL divergence between $\nu^+$ and $\nu^-$ 
   satisfies 
\begin{align}
   \dkl(\nu^+ \| \nu^-) = 2a^2 =2\Delta^2
 \label{eq:kl-normal3}
\end{align}
    (see Theorem~\ref{prop-normal}).

  Suppose we have a $\delta$-PAC algorithm for the problem for $0<\delta<1/3$:  
  Let $\mathcal{\mathcal{E}^+}$ (resp.\  $\mathcal{\mathcal{E}^-}$) 
   denote the event that algorithm determines ``$A$ is non-redundant (resp.\ redundant).'' 
  Let $\mathbb{P}_{\nu^+}(\mathcal{E}^+)$ (resp. $\mathbb{P}_{\nu^-}(\mathcal{E}^+)$) 
    denote the probability of $\mathcal{E}^+$ 
    under $\nu^+$ (resp.  $\nu^-$), 
    then the PAC algorithm must satisfy both of 
\begin{align*}
  \mathbb{P}_{\nu^+}(\mathcal{E}^+) \geq 1-\delta 
  \qquad\mbox{and}\qquad
  \mathbb{P}_{\nu^-}(\mathcal{E}^+) \leq \delta
\end{align*}
 which implies 
\begin{align}
  \mathbb{P}_{\nu^+}(\mathcal{E}^+)-  
  \mathbb{P}_{\nu^-}(\mathcal{E}^+) \geq 1-2\delta. 
\label{eq:pdiff3}
\end{align}
 Notice that 
\begin{align}
  \mathbb{P}_{\nu^+}(\mathcal{E}^+)-  
  \mathbb{P}_{\nu^-}(\mathcal{E}^+) 
  \leq \sup_{A \in E} | \mathbb{P}_{\nu^+} (A)-  \mathbb{P}_{\nu^-}(A) |
 = \dtv (\mathbb{P}_{\nu^+},\mathbb{P}_{\nu^-})
\label{eq:dtv3}
\end{align}
holds. 
By the Bretagnolle–Huber inequality Theorem~\ref{thm:BHineq}, 
\begin{align}
\dtv (\mathbb{P}_{\nu^+},\mathbb{P}_{\nu^-}) \leq 1-2\exp(-\dkl(\mathbb{P}_{\nu^+} \| \mathbb{P}_{\nu^-}))
\label{eq:BH3}
\end{align}
holds. 
By using the chain rule, we can prove that 
\begin{align}
\dkl(\mathbb{P}_{\nu^+} \| \mathbb{P}_{\nu^-}) = \tau \dkl(\nu^+ \| \nu^-) = 2 \tau \Delta^2
\label{eq:dklx3}
\end{align}
where the last equality follows \eqref{eq:kl-normal3}. 
Thus \eqref{eq:pdiff3}--\eqref{eq:dklx3} imply 
\begin{align*}
 1-2\delta \leq 1-2\exp(-2 \tau \Delta^2)
\end{align*}
and hence 
\begin{align*}
 \tau \geq \frac{-\ln \delta}{2\Delta^2}. 
\end{align*}
\end{proof}

\section{Supplemental Proofs of Section~\ref{sec:upper}}\label{apx:upper}
\begin{proof}[Proof of Lemma~\ref{lem:rankAj}]
Note that $\rank(A) = n-1$ implies that $\dim(\ker(A))=1$. 
 Let $\Vec{c} \in \ker(A) \setminus \{\Vec{0}\}$, i.e., 
$\Vec{c} \neq \Vec{0}$ and $\Vec{c}^{\top}A=\Vec{0}^{\top}$. 

($\Leftarrow$)
 Since $\sum_{i=1}^n x_i \neq 0$, 
 let $\Vec{c}' = \frac{\Vec{c}}{\sum_{i=1}^n c_i}$. 
 Then, $\Vec{c}'$ is the unique solution of $\Vec{x}^{\top}A=0$ and $x_1 + \cdots + x_n = 1$. 
 This implies that $\Vec{1}$ is independent of the column space of $A$.  
 Now, the claim is easy from a standard argument of linear algebra.

($\Rightarrow$) Consider the contraposition: 
 if $\Vec{x}^{\top}A \neq \Vec{0}$ or $\sum_{i=1}^n x_i = 0$ holds for all $\Vec{x}\neq\Vec{0}$ 
 then $\rank(A_j) \neq n$ for some $j \in \{1,\ldots,n\}$. 
 It is trivial from $\Vec{c}^{\top}(\Vec{1},A) = \Vec{0}^{\top}$. 
\end{proof}

\section{Proofs of Section~\ref{sec:low}}\label{apx:low}
\subsection{Preliminary}
 To begin with, we construct a bad example. 
Let $Q=Q(n,\kappa,\mys) = (q_{ij})$ 
  for $0< \kappa \ll 1$ and $|\mys| \ll \kappa$
  be an $n \times n$ matrix for an odd $n$ satisfying $n \geq 5$  given by 
\begin{align}
 q_{ij} = \begin{cases}
  0 & \mbox{if $i=j$} \\
  \kappa & \mbox{if $1 \leq  j - i  \leq \frac{n-1}{2} $} \\
 -\kappa & \mbox{if $\frac{n+1}{2} \leq j - i \leq n-2$} \\
 -2\kappa+\mys & \mbox{if $(i,j)=(1,n)$} \\
 -q_{ji} & \mbox{if $i > j$} \\
 \end{cases}
\end{align}
for $(i,j) \in n^2$.  
For instance, $Q=Q(n,\kappa,\mys)$ is given by 
\begin{align*}
Q=\kappa\begin{pmatrix}
0 & 1 & 1 &1&-1& -1 & -2+\frac{\mys}{\kappa} \\
-1 & 0 & 1 & 1 &1&-1& -1 \\
-1 & -1 & 0 & 1 & 1 &1&-1\\
-1 & -1 & -1 & 0 & 1 &1&1\\
1&-1 & -1 & -1 & 0 & 1 &1\\
1&1&-1 & -1 & -1 & 0 & 1 \\
2-\frac{\mys}{\kappa} & 1 &1&-1& -1 & -1 & 0
\end{pmatrix}
\end{align*}
for $n=7$. 
As we will see in Section~\ref{apx:low2}, 
 $\det(Q_k) = \kappa^n \left( (n-4)\left(\frac{\mys}{\kappa}\right)^2  + 4\frac{|\mys|}{\kappa} \right)$, and 
 $\varphi(Q) \simeq \frac{2n-8}{|\mys|}$ asymptotic to $|\mys| \to 0$ which is almost independent of $\kappa$, interestingly. 

This section establishes the following lemma. 
\begin{lemma}\label{lem:Q-red}
 For any odd $n \geq 5$, 
  $Q$ is non-redundant if $0 < \mys < 2\kappa$, otherwise redundant. 
\end{lemma}
 Lemma~\ref{lem:Q-red} immediately follows from Lemmas~\ref{lem:Q-rank} and \ref{lem:Q-sol} appearing below. 
\begin{lemma}\label{lem:Q-rank}
 $\rank(Q) = n-1$ unless $\mys \in \{0,2\kappa\}$. 
\end{lemma}
 The lemma is proved by some artificial and systematic elementary row operations, 
  but it is quite lengthy and we omit the detail. 
 Instead, 
  the readers can be confirmed with a supplemental python program proof\_sol.py Lemma~\ref{lem:Q-rank} for small $n$. 

Next, we are concerned with the solution of $\Vec{x}^{\top}Q = \Vec{0}$. 
\begin{lemma}\label{lem:Q-sol}
Let $\Vec{x}$ be 
\begin{align}
 x_i = \begin{cases}
  \kappa & (\mbox{if $i\in \{1,n\}$})\\
  2\kappa-\mys & (\mbox{if $i=\frac{n+1}{2}$}) \\
 \mys & (\mbox{otherwise}).
 \end{cases}
\label{eq:Q-sol}
\end{align}
Then $\Vec{x}^{\top} Q = \Vec{0}^{\top}$. 
\end{lemma}

 Before the proof of Lemma~\ref{lem:Q-sol}, 
  we see an example in the case of $n=7$. 
 The vector $\Vec{x}$ given by \eqref{eq:Q-sol} is described as 
\begin{align*}
\Vec{x}^{\top}=\begin{pmatrix}
\kappa & \mys & \mys & 2\kappa-\mys & \mys & \mys & \kappa
\end{pmatrix}
\end{align*}
then 
\begin{align*}
\Vec{x}^{\top} Q
&=\begin{pmatrix}
\kappa & \mys & \mys & 2\kappa-\mys & \mys & \mys & \kappa
\end{pmatrix}
\begin{pmatrix}
0 & 1 & 1 &1&-1& -1 & -2+\frac{\mys}{\kappa} \\
-1 & 0 & 1 & 1 &1&-1& -1 \\
-1 & -1 & 0 & 1 & 1 &1&-1\\
-1 & -1 & -1 & 0 & 1 &1&1\\
1&-1 & -1 & -1 & 0 & 1 &1\\
1&1&-1 & -1 & -1 & 0 & 1 \\
2-\frac{\mys}{\kappa} & 1 &1&-1& -1 & -1 & 0
\end{pmatrix} \\
&= \begin{pmatrix}
0 & 0 & 0 & 0 & 0 & 0 & 0
\end{pmatrix}
\end{align*}
 and we see  $\Vec{x}^{\top}Q = \Vec{0}$. 
 Now we prove Lemma~\ref{lem:Q-sol}. 
\begin{proof}[Proof of Lemma~\ref{lem:Q-sol}]
We prove $(\Vec{x}^{\top} Q)_j = 0$ for $j=1,\ldots,n$. 
Firstly,  we remark that 
\begin{align*}
\textstyle
 \sum_{i=2}^{\frac{n-1}{2}} q_{ij} + \sum_{\frac{n+3}{2}}^{n-1} q_{ij}  = 
 \begin{cases}
 0 & (\mbox{for $j=1,\frac{n+1}{2}$ and $n$})\\
-\kappa & (\mbox{for $2 \leq j \leq \frac{n-1}{2},n$})\\
 \kappa & (\mbox{for $\frac{n+3}{2} \leq j \leq n-1$})
 \end{cases}
\end{align*}
holds. 
Then 
\begin{align*}
 \textstyle  (\Vec{x}^{\top} A)_1 
  &\textstyle = - \kappa x_{\frac{n+1}{2}} + (2\kappa-\mys)x_n + \sum_{i=2}^{\frac{n-1}{2}} a_{ij}x_i  + \sum_{\frac{n+3}{2}}^{n-1} a_{ij}x_i  \\
  &\textstyle = -\kappa(2\kappa-\mys) + (2\kappa-\mys)\kappa + \mys \left(\sum_{i=2}^{\frac{n-1}{2}} a_{ij}  + \sum_{\frac{n+3}{2}}^{n-1} a_{ij} \right)  =0 
\end{align*}
hold since $x_1=x_n=\kappa$, $x_{\frac{n+1}{2}} = 2\kappa-\mys$ and $x_i = \mys$ for other $i$.
Similarly, we have 
\begin{align*}
 &\textstyle (\Vec{x}^{\top} A)_n 
  =  (-2\kappa+\mys)x_1 + \kappa x_{\frac{n+1}{2}}  + \mys \left(\sum_{i=2}^{\frac{n-1}{2}} a_{ij}  + \sum_{\frac{n+3}{2}}^{n-1} a_{ij} \right) =0 \\
 &\textstyle (\Vec{x}^{\top} A)_{\frac{n+1}{2}}
  = - \kappa x_1 + \kappa x_2 + \mys \left(\sum_{i=2}^{\frac{n-1}{2}} a_{ij}  + \sum_{\frac{n+3}{2}}^{n-1} a_{ij} \right) =0 \\
 &\textstyle (\Vec{x}^{\top} A)_j 
  =  \kappa  x_1 - \kappa x_{\frac{n+1}{2}} + \kappa x_n + \mys \left(\sum_{i=2}^{\frac{n-1}{2}} a_{ij}  + \sum_{\frac{n+3}{2}}^{n-1} a_{ij} \right) =0 
 && \mbox{for $2 \leq j \leq \frac{n-1}{2}$}  \\
 &\textstyle (\Vec{x}^{\top} A)_j 
  = -\kappa x_1 + \kappa x_{\frac{n+1}{2}} - \kappa x_n + \mys \left(\sum_{i=2}^{\frac{n-1}{2}} a_{ij}  + \sum_{\frac{n+3}{2}}^{n-1} a_{ij} \right) =0 
 && \mbox{for $\frac{n+1}{2} \leq j \leq n-1$}  
\end{align*}
hold. 
\end{proof}
You may be confirmed with a supplemental python program proof\_sol.py Proposition~\ref{lem:Q-sol} for small $n$.

\subsection{Proof of Theorem~\ref{thm:ETlow}}
 We will use the technique developed by \citet{KCG16} in proof of Theorem~\ref{thm:ETlow}, 
  where we use the following two theorems\footnote{
  We also give Theorem~\ref{thm:ETlow-weak} in the next section for an alternative of Theorem~\ref{thm:ETlow}, 
   where the proof of Theorem~\ref{thm:ETlow-weak} might be more familiar to some readers. 
  }.  
\begin{theorem}[cf.~Lem.~19 in \cite{KCG16}]\label{lem:Kaufmann19}
 Let $\nu$ and $\nu'$ be two of bandit models, 
  where observations are iid respectively according to density functions $f_{\nu}$ and $f_{\nu'}$.  
 Let $L(t)$ be the log-likelihood ratio of the observations up to time $t$ under algorithm $\mathcal{A}$ 
  which is given by 
\begin{align}
 L(t) = \sum_{s=1}^t \log\left(\frac{f_{\nu}(x_s)}{f_{\nu'}(x_s)}\right). 
\end{align}
 Let $T$ be an almost surely finite stopping time with respect to $\mathcal{F}_t$. 
 Then, 
\begin{align}
 E_{\nu}[L(T)] \geq d(\mathbb{P}_{\nu}(\mathcal{E}),\mathbb{P}_{\nu'}(\mathcal{E})) 
\end{align}
 holds 
  for every event $\mathcal{E} \in \mathcal{F}_T$,  
where $d(p,q)= p\log{\frac{p}{q}}+(1-p)\log{\frac{1-p}{1-q}}$.  
\end{theorem}
\begin{theorem}[cf.~(3) in \cite{KCG16}]\label{lem:2.4}
Let $d(p,q)= p\log{\frac{p}{q}}+(1-p)\log{\frac{1-p}{1-q}}$ then 
\begin{align*}
d(p,1-p) \geq \log \frac{5}{12p}
\end{align*}
 holds for any $p \in[0,1]$. 
\end{theorem}

Now, we prove Theorem~\ref{thm:ETlow}. 
\begin{proof} of Theorem~\ref{thm:ETlow}. 
 Let $Q^+ = (q^+_{ij}) = Q(n,\kappa,\alpha)$ and let $Q^- = (q^-_{ij}) = Q(n,\kappa,-\alpha)$, 
  where we remark that $q^+_{ij} = q^-_{ij}$ unless $(i,j) = (1,n)$ for any $i<j$. 
 Note that $Q^+$ is non-redundant while $Q^-$ is redundant by Lemma~\ref{lem:Q-red}. 
 We are concerned with the following decision problem under the dueling bandit setting: 
   given an unknown matrix $A=(a_{ij})$ such that $A \in \{Q^+,Q^-\}$ and 
   decide whether $A$ is redundant or not 
   by observing the results of duels. 
 Here,  the result $X_{ij}$ ($i<j$) of duels is deterministic unless $(i,j) = (1,n)$
   so that $X_{ij} = a_{ij}$  at any time, and  
 the result $X_{1n}$ follows the normal distribution $\mathcal{N}(a_{1n}, 1)$ with mean $a_{1n}$ and variance $1$.  
 For convenience, 
  let $\nu^+ = \mathcal{N}(q^+_{1n}, 1)$ and $\nu^- = \mathcal{N}(q^-_{1n}, 1)$,  
  thus $X_{1n}$ follows $\nu^+$ if $A=Q^+$, otherwise $\nu^-$. 
 We note that the KL divergence between $\nu^+$ and $\nu^-$ 
   satisfies 
\begin{align}
   \dkl(\nu^+ \| \nu^-) = 2\alpha^2
 \label{eq:kl-normal}
\end{align}
    by Theorem~\ref{prop-normal}. 

  Suppose we have a $(\alpha,\delta)$-PAC algorithm for the problem for $0<\delta<1/3$.   
  Let $\mathcal{\mathcal{E}^+}$ (resp.\  $\mathcal{\mathcal{E}^-}$) 
   denote the event that algorithm determines ``$A$ is non-redundant (resp.\ redundant).'' 
  Let $\mathbb{P}_{\nu^+}(\mathcal{E}^+)$ (resp. $\mathbb{P}_{\nu^-}(\mathcal{E}^+)$) 
    denote the probability of $\mathcal{E}^+$ 
    under $\nu^+$ (resp.  $\nu^-$), 
    then the PAC algorithm must satisfy both of 
\begin{align}
  \mathbb{P}_{\nu^+}(\mathcal{E}^+) \geq 1-\delta 
  \qquad\mbox{and}\qquad
  \mathbb{P}_{\nu^-}(\mathcal{E}^+) \leq \delta. 
\label{cond:PAC}
\end{align}
 The following arguments follows the technique of Kaufman et al.~\cite{KCG16}. 
 Firstly, let $d(p,q)= p\log{\frac{p}{q}}+(1-p)\log{\frac{1-p}{1-q}}$ then 
\begin{align}
d\left(\mathbb{P}_{\nu^+}(\mathcal{E}^+), \mathbb{P}_{\nu^-}(\mathcal{E}^+) \right) 
&\geq d\left(1-\delta,\delta \right) 
&&(\mbox{by \eqref{cond:PAC}})\nonumber \\
&\geq \log{\frac{5}{12\delta}} 
&&(\mbox{by Theorem~\ref{lem:2.4}})
\label{eq:dlog}
\end{align} 
holds. 
 Next, let $\tau$ be a positive integer valued random variable 
    denoting the stopping time of the PAC algorithm. 
 We are concerned with the expectation of $\tau$ under $\nu^+$, 
  which is denoted by $\mathbb{E}_{\nu^+}[\tau]$. 
By Theorem~\ref{lem:Kaufmann19} 
\begin{align}
\mathbb{E}_{\nu^+}[L(\tau)] 
\geq d\left(\mathbb{P}_{\nu^+}(\mathcal{E}^+), \mathbb{P}_{\nu^-}(\mathcal{E}^+) \right) 
\label{eq:elt1}
\end{align}
 holds where 
\begin{align*}
L(t) = \log\left(\frac{f_{\nu+}(x_1,\ldots,x_t)}{f_{\nu-}(x_1,\ldots,x_t)}\right) 
=\log\left(\left(\frac{f_{\nu+}(x)}{f_{\nu-}(x)}\right)^t\right)
= t \log\left(\frac{f_{\nu+}(x)}{f_{\nu-}(x)}\right)
\end{align*}
for $0 <t <1$. 
Notice that
\begin{align}
\mathbb{E}_{\nu^+}[L(\tau)] 
&= \mathbb{E}_{\nu^+}[\tau] \mathbb{E}_{\nu^+}\left[\log\left(\frac{f_{\nu+}(x)}{f_{\nu-}(x)}\right) \right] 
&& (\mbox{by Wald's Lemma})\nonumber\\
&= \mathbb{E}_{\nu^+}[\tau] \dkl(\nu^+ \| \nu^-) && (\mbox{by definition of $\dkl$}) \nonumber\\
&= \mathbb{E}_{\nu^+}[\tau] 2\alpha^2 && (\mbox{by \eqref{eq:kl-normal}})
\label{eq:elt2}
\end{align}
holds. 
 Then, 
\begin{align}
\mathbb{E}_{\nu^+}[\tau]
&=\frac{1}{2\alpha^2}\mathbb{E}_{\nu^+}[L(\tau)] 
&&(\mbox{by \eqref{eq:elt2}})  \nonumber \\
&\geq \frac{1}{2\alpha^2}d\left(\mathbb{P}_{\nu^+}(\mathcal{E}^+), \mathbb{P}_{\nu^-}(\mathcal{E}^+) \right) 
&&(\mbox{by \eqref{eq:elt1}})\nonumber \\
&\geq \frac{1}{2\alpha^2}\log{\frac{5}{12\delta}} 
&&(\mbox{by \eqref{eq:dlog}})
\end{align} 
and we obtain the claim. 
\end{proof}

\subsection{Alternative to Theorem~\ref{thm:ETlow}}\label{apx:ETlow-weak}
 This section gives Theorem~\ref{thm:ETlow-weak} as an alternative to Theorem~\ref{thm:ETlow}, 
   where the proof might be more familiar to some readers. 
\begin{theorem}\label{thm:ETlow-weak}
 Let $\alpha$ and $\delta$ be fixed parameters respectively satisfying $0 < \alpha \ll 1$ and $0 < \delta \ll 1$. 
 Let $\tau$ denote the running time of 
   an arbitrary $(\alpha,\delta)$-PAC algorithm 
   that identifies whether an arbitrarily given $A$ is non-redundant.  
 Then, 
  $\tau\geq \frac{1}{2\alpha^2}\log{\frac{1}{\delta}}$. 
\end{theorem}

\begin{proof}[Proof of Theorem~\ref{thm:ETlow-weak}]
 Let $Q^+ = (q^+_{ij}) = Q(n,\kappa,\alpha)$ and let $Q^- = (q^-_{ij}) = Q(n,\kappa,-\alpha)$, 
  where we remark that $q^+_{ij} = q^-_{ij}$ unless $(i,j) = (1,n)$ for any $i<j$. 
 Note that $Q^+$ is non-redundant while $Q^-$ is redundant by Lemma~\ref{lem:Q-red}. 
 We are concerned with the following decision problem under the dueling bandit setting; 
   given an unknown matrix $A=(a_{ij})$ such that $A \in \{Q^+,Q^-\}$ and 
   decide whether $A$ is redundant or not 
   by observing the results of duels. 
 Here,  the result $X_{ij}$ ($i<j$) of duels is deterministic unless $(i,j) = (1,n)$
   so that $X_{ij} = a_{ij}$  at any time, and  
 the result $X_{1n}$ follows the normal distribution $\mathcal{N}(a_{1n}, 1)$ with mean $a_{1n}$ and variance $1$.  
 For convenience, 
  let $\nu^+ = \mathcal{N}(q^+_{1n}, 1)$ and $\nu^- = \mathcal{N}(q^-_{1n}, 1)$,  
  thus $X_{1n}$ follows $\nu^+$ if $A=Q^+$, otherwise $\nu^-$. 
 We note that the KL divergence between $\nu^+$ and $\nu^-$ 
   satisfies 
\begin{align}
   \dkl(\nu^+ \| \nu^-) = 2\alpha^2
 \label{eq:kl-normal2}
\end{align}
    by Theorem~\ref{prop-normal}. 

  Suppose we have a $(\alpha,\delta)$-PAC algorithm for the problem for $0<\delta<1/3$:  
  Let $\mathcal{\mathcal{E}^+}$ (resp.\  $\mathcal{\mathcal{E}^-}$) 
   denote the event that algorithm determines ``$A$ is non-redundant (resp.\ redundant).'' 
  Let $\mathbb{P}_{\nu^+}(\mathcal{E}^+)$ (resp. $\mathbb{P}_{\nu^-}(\mathcal{E}^+)$) 
    denote the probability of $\mathcal{E}^+$ 
    under $\nu^+$ (resp.  $\nu^-$), 
    then the PAC algorithm must satisfy both of 
\begin{align*}
  \mathbb{P}_{\nu^+}(\mathcal{E}^+) \geq 1-\delta 
  \qquad\mbox{and}\qquad
  \mathbb{P}_{\nu^-}(\mathcal{E}^+) \leq \delta
\end{align*}
 which implies 
\begin{align}
  \mathbb{P}_{\nu^+}(\mathcal{E}^+)-  
  \mathbb{P}_{\nu^-}(\mathcal{E}^+) \geq 1-2\delta. 
\label{eq:pdiff}
\end{align}
 Notice that 
\begin{align}
  \mathbb{P}_{\nu^+}(\mathcal{E}^+)-  
  \mathbb{P}_{\nu^-}(\mathcal{E}^+) 
  \leq \sup_{A \in E} | \mathbb{P}_{\nu^+} (A)-  \mathbb{P}_{\nu^-}(A) |
 = \dtv (\mathbb{P}_{\nu^+},\mathbb{P}_{\nu^-})
\label{eq:dtv}
\end{align}
holds. 
By the Bretagnolle–Huber inequality Theorem~\ref{thm:BHineq}, 
\begin{align}
\dtv (\mathbb{P}_{\nu^+},\mathbb{P}_{\nu^-}) \leq 1-2\exp(-\dkl(\mathbb{P}_{\nu^+} \| \mathbb{P}_{\nu^-}))
\label{eq:BH}
\end{align}
holds. 
By using the chain rule, we can prove that 
\begin{align}
\dkl(\mathbb{P}_{\nu^+} \| \mathbb{P}_{\nu^-}) = \tau \dkl(\nu^+ \| \nu^-) = 2 \tau \alpha^2
\label{eq:dklx}
\end{align}
where the last equality follows \eqref{eq:kl-normal2}. 
Thus \eqref{eq:pdiff}--\eqref{eq:dklx} imply 
\begin{align*}
 1-2\delta \leq 1-2\exp(-2 \tau \alpha^2)
\end{align*}
and hence 
\begin{align*}
 \tau \geq \frac{-\ln \delta}{2\alpha^2}. 
\end{align*}
\end{proof}

\subsection{Proof of Theorem~\ref{thm:low-varphi}}\label{apx:low2}
 This section proves Theorem~\ref{thm:low-varphi} using Theorem~\ref{thm:ETlow}. 
 For the purpose, 
  we give an upper bound of  $\varphi(Q)
  = \underset{j \in [n]}{\max} \ \! \underset{i \in [n]}{\max} \ \! 
           \left| \left(  \Vec{1}^{\top} Q_j^{-1} -\Vec{\pi} \right)_i \right| $ given in Proposition~\ref{prop:up_varphiQ}
   considering a lower bound of $|\det(Q_k)|$ given in Proposition~\ref{prop:Qdet} and 
   an upper bound of the absolute values of cofacters of $Q_k$ given in Proposition~\ref{prop:Qcof}. 
For convenience, 
  let $Q^{\mathrm{o}} = \frac{1}{\kappa}Q$, i.e., $Q = \kappa Q^{\mathrm{o}} $ 
 in the following arguments. 
\begin{proposition}\label{prop:Qdet}
If $\frac{|\mys|}{\kappa} \leq \frac{1}{n}$ then 
$|\det(Q_k)| \geq 3 \kappa^n  \frac{|\mys|}{\kappa} $ for any $k=1,\ldots,n$. 
\end{proposition}
\begin{proof}[Proof sketch]
Notice that $\det(Q_k) = \kappa^n\det(Q^{\mathrm{o}}_k)$. Then, we are concerned with $\det(Q^{\mathrm{o}}_k)$. 
We can prove by elementary row operations that 
 $\det(Q^{\mathrm{o}}_1) = \det(Q^{\mathrm{o}}_n) = (n-4)\frac{\mys}{\kappa} + 4$ and 
 $\det(Q^{\mathrm{o}}_{\frac{n+1}{2}}) = (n-4)\left(\frac{\mys}{\kappa}\right)^2 + 2(n-6)\frac{\mys}{\kappa} + 8$. 
 On condition that $\frac{|\mys|}{\kappa} \leq \frac{1}{n}$, 
  it is not difficult to see that 
  $|\det(Q^{\mathrm{o}}_1)| = | \det(Q^{\mathrm{o}}_n)| \geq |-1+4| = 3$ and 
  $|\det(Q^{\mathrm{o}}_{\frac{n+1}{2}})| \geq |0 - 2 + 8| = 6$, 
  which implies the claim for $k = 1,\frac{n+1}{2}$ and~$n$. 
 Consider the other cases of $k$. 
 Let $m_{ij}$ denote $(i,j)$-cofactor of $Q^{\mathrm{o}}_k$, 
  then notice that $\det(Q^{\mathrm{o}}_k) = \sum_{i=1}^n (Q^{\mathrm{o}}_k)_{ik} m_{ik}= \sum_{i=1}^n m_{ik}$ 
  by the cofactor expansion along the $k$-th column 
  since $(Q^{\mathrm{o}}_k)_{ik} = 1$ for $i=1,\ldots,n$. 
 We can prove in any case of $k \in [n] \setminus \{1,2,\frac{n+3}{2}\}$ that 
\begin{align*}
m_{ik} = \begin{cases}
 \frac{\mys}{\kappa} & (\mbox{if $i\in\{1,n\}$}) \\
 -\left(\frac{\mys}{\kappa}\right)^2 +2\frac{\mys}{\kappa}  & (\mbox{if $i=\frac{n+1}{2}$}) \\
 \left(\frac{\mys}{\kappa}\right)^2 & (\mbox{otherwise})
 \end{cases}
\end{align*}
hold for $i=1,\ldots,n$. 
Then, 
$\det(Q^{\mathrm{o}}_k) = \sum_{i=1}^n m_{ik} =  (n-4)\left(\frac{\mys}{\kappa}\right)^2  + 4\frac{\mys}{\kappa}$ 
 for any of those $k$.
Since the assumption $\frac{|\mys|}{\kappa} \leq \frac{1}{n}$, 
we obtain $|\det(Q^{\mathrm{o}}_k)| \geq \left|-(n-4)\frac{1}{n} \frac{\mys}{\kappa} + 4\frac{\mys}{\kappa} \right|$. 
Now the claim is easy. 
\end{proof}
\begin{proposition}\label{prop:Qcof}
 Let $m_{ijk}$ denote the $(i,j)$-cofactor of $Q_k$. 
 If $\frac{|\mys|}{\kappa} \leq \frac{1}{n}$ 
 then $\max_{i,j,k}|m_{ij}^{k}| \leq 4\kappa^{n-1}n$ for $k =5,7,9,\ldots,19$. 
\end{proposition}
\begin{proof}
Let $m^{\mathrm{o}}_{ijk}$ denote the $(i,j)$-cofactor of $Q^{\mathrm{o}}_k$. 
Notice that $m_{ijk} = \kappa^{n-1}m^{\mathrm{o}}_{kij}$
By our calculation with proof\_cof.py, 
\begin{align*}\textstyle
\max_{i,j,k}|m^{\mathrm{o}}_{ijk}| \leq 
(2n-8)\left(\frac{\mys}{\kappa}\right)^2 + (4n-16)\frac{|\mys|}{\kappa} + (4n-16). 
\end{align*}
Since $\frac{|\mys|}{\kappa} \leq \frac{1}{n}$, the claim is easily confirmed. 
\end{proof}

As a consequence of Propositions~\ref{prop:Qdet} and \ref{prop:Qcof}, we obtain 
\begin{proposition}\label{prop:up_varphiQ}
If $\frac{|\mys|}{\kappa}  \leq \frac{1}{n}$ then 
 $\varphi(Q) \leq  \frac{4n^2+1}{3|\mys|} $
for $n=5,7,\ldots,19$.
\end{proposition}
\begin{proof}[Proof of Proposition~\ref{prop:up_varphiQ}]
\begin{align*}
\varphi(Q)
  &= \underset{j \in [n]}{\max} \ \! \underset{i \in [n]}{\max} \ \! 
           \left| \left(  \Vec{1}^{\top} Q_j^{-1} -\Vec{\pi} \right)_i \right| \nonumber\\
   &\leq \underset{j \in [n]}{\max} \ \! \underset{i \in [n]}{\max} \ \! 
           \left| \left(  \Vec{1}^{\top} Q_j^{-1} \right)_i + 1 \right| \nonumber\\
   &= \underset{j \in [n]}{\max} \ \! \underset{i \in [n]}{\max} \ \! 
           \left| \left(  \sum_{l=1}^n \frac{m_{il}^j}{\det(Q_j)} \right)_i + 1 \right| \nonumber\\
   &\leq \underset{j \in [n]}{\max} \ \! \underset{i \in [n]}{\max} \ \! 
            \left(  \sum_{l=1}^n \frac{\left|m_{il}^j\right|}{\left|\det(Q_j)\right|} \right)_i + 1  \nonumber\\
   &\leq   \sum_{l=1}^n \frac{4r^{n-1}n}{3r^n\frac{|\mys|}{r}}  + 1 
       &&(\mbox{by Propositions~\ref{prop:Qdet} and \ref{prop:Qcof}})    \nonumber\\
   &=   \sum_{l=1}^n \frac{4n}{3|\mys|}  + 1   \nonumber\\
   &=   \frac{4n^2+3|\mys|}{3|\mys|}  \nonumber\\
   &\leq   \frac{4n^2+1}{3|\mys|} 
       &&(\mbox{since $|\mys| \leq \frac{r}{n} \leq \frac{1}{5}$})
\end{align*}
and we obtain the claim. 
\end{proof}
By our calculation, we can observe 
 $\varphi(Q)
           \simeq \frac{2n-8}{|\mys|}$ for $n \geq 7$ asymptotic to $|\mys| \to 0$. 

\begin{proof}[Proof of Theorem~\ref{thm:low-varphi}]
 Let $Q^+ = Q(n,\kappa,\alpha)$ and $Q^- = Q(n,\kappa,-\alpha)$, and let $A \in \{Q^+,Q^-\}$. 
 Since $\varphi(A) \leq  \frac{4n^2+1}{3\alpha} $ by Proposition \ref{prop:up_varphiQ}, 
  we have
\begin{align}
   \frac{1}{\alpha} 
    \geq \frac{3}{4n^2+1}\varphi(A). 
\label{20250503a}
\end{align}
By Theorem~\ref{thm:ETlow}, 
\begin{align*}
 \mathbb{E}_A{[\tau]}
  & \geq \frac{1}{2\alpha^2}\log{\frac{5}{12\delta}}\\
  & \geq \frac{1}{2} \left(\frac{3}{4n^2+1}\right)^2\varphi(A)^2\log{\frac{5}{12\delta}} 
  &&(\mbox{by \eqref{20250503a}})
  \\
  & \geq \frac{1}{2} \left(\frac{1}{2n^2}\right)^2\varphi(A)^2\log{\frac{5}{12\delta}} 
  &&(\mbox{since $\frac{3}{4n^2+1} \geq \frac{1}{2n^2}$}) \\
  & = \frac{1}{8n^4} \varphi(A)^2\log{\frac{5}{12\delta}}
\end{align*} 
and we obtain the claim. 
\end{proof}

\subsection{Remark on $\pi_{\min}$ of $Q(n,\kappa,\mys)$}
 One might think that $1/\pi_{\min}$ is the true leading term of (the lower bound of) the sample complexity. 
 The following fact claims that $\pi_{\min}$ for $Q(n,\kappa,\mys)$ is not very small. 
\begin{proposition}\label{prop:Q-pi}
 If $0 < \mys \leq \frac{\kappa}{n}$ then 
  $Q(n,\kappa,\mys)$ has the unique Nash equilibrium $\Vec{\pi}=(\pi_1,\ldots,\pi_n)$ and 
\begin{align*}
 \pi_{\min}  \geq \frac{\mys}{5\kappa}
\end{align*} 
 holds where $\pi_{\min} = \min_i \pi_i$. 
\end{proposition}
\begin{proof}
 It is easy from Lemmas~\ref{lem:Q-sol} and \ref{lem:Q-red} to see that 
 $\Vec{\pi} = \frac{1}{4\kappa+(n-4)\mys} \Vec{x}$ 
 is the unique Nash equilibrium of $Q(n,\kappa,\mys)$ if $0 < \mys < 2\kappa $. 
 Recalling \eqref{eq:Q-sol}, 
  notice that $\min\{\kappa,2\kappa-\mys,\mys\} = \mys$ holds 
  if $0 < \mys \leq \frac{\kappa}{n}$.  
 Then, 
\begin{align*}
 \pi_{\min} = \frac{\mys}{4\kappa+(n-4)\mys} \geq \frac{\mys}{5\kappa}
\end{align*} 
 holds, and we obtain the claim. 
\end{proof}

 By Proposition~\ref{prop:Q-pi}, 
  we observe
  $\pi_{\min} \geq \frac{1}{5n}$ 
  holds for the Nash equilibrium of $A= Q(n,\kappa,\frac{\kappa}{n})$. 
 On the other hand $\varphi(A) \to \infty$ as $\alpha \to 0$ (by setting $\kappa \to 0$). 
 This means that $\varphi(A)^2$ term of the lower bound 
   is not replaced with a function of $1/\pi_{\min}^2$.

\end{document}